\documentclass[sigplan,10pt]{acmart}
\renewcommand\footnotetextcopyrightpermission[1]{} 

\usepackage[normalem]{ulem}

\usepackage{times}
\usepackage{textcomp}
\usepackage{adjustbox}
\usepackage{xspace}
\usepackage{amsmath}
\usepackage{amsthm}
\usepackage{float}
\usepackage{algpseudocode}
\usepackage{graphicx}
\usepackage{caption}
\DeclareCaptionType{copyrightbox}
\usepackage{wrapfig}
\usepackage{amssymb}
\usepackage{microtype}
\usepackage{xcolor}
\usepackage{multirow}
\usepackage{multicol}
\usepackage{balance}
\usepackage{nowidow}
\usepackage{import}
\usepackage{mathtools}
\usepackage{listings}
\usepackage{url}
\usepackage{hhline}
\usepackage{tabularx}
\usepackage{subfig}
\usepackage[mathscr]{euscript}
 \let\mathscr\relax
\usepackage{bigints}
\usepackage{adjustbox}
\usepackage{array}
\usepackage{centernot}
\usepackage{esvect}
\usepackage{xparse}

\newcolumntype{x}[1]{>{\centering\arraybackslash\hspace{0pt}}p{#1}}
\newcolumntype{L}[1]{>{\raggedright\let\newline\\\arraybackslash\hspace{0pt}}m{#1}}
\newcolumntype{C}[1]{>{\centering\let\newline\\\arraybackslash\hspace{0pt}}m{#1}}
\newcolumntype{R}[1]{>{\raggedleft\let\newline\\\arraybackslash\hspace{0pt}}m{#1}}

\usepackage{color, colortbl}
\usepackage{color}

\definecolor{lightgray}{gray}{0.92}
\definecolor{applegreen}{rgb}{0.13, 0.67, 0.8}
\makeatletter 
\renewcommand{\@seccntformat}[1]{\csname the#1\endcsname\ \ }
\makeatother

\usepackage{algorithm}

\newcommand{\rfig}[1]{Figure~\ref{#1}}
\newcommand{\rtab}[1]{Table~\ref{#1}}

\newcommand{\rsec}[1]{\S\ref{#1}}

\newcommand{\adj}[1]{\text{adj}(#1)}

\newsavebox{\fminipagebox}
\NewDocumentEnvironment{fminipage}{m O{\fboxsep}}
 {\par\kern#2\noindent\begin{lrbox}{\fminipagebox}
  \begin{minipage}{#1}\ignorespaces}
 {\end{minipage}\end{lrbox}%
  \makebox[#1]{%
    \kern\dimexpr-\fboxsep-\fboxrule\relax
    \fbox{\usebox{\fminipagebox}}%
    \kern\dimexpr-\fboxsep-\fboxrule\relax
  }\par\kern#2
 }

\algblock{ParFor}{EndParFor}
\algnewcommand\algorithmicparfor{\textbf{par-for}}
\algnewcommand\algorithmicpardo{\textbf{do}}
\algnewcommand\algorithmicendparfor{\textbf{end\ par-for}}
\algrenewtext{ParFor}[1]{\algorithmicparfor\ #1\ \algorithmicpardo}
\algrenewtext{EndParFor}{\algorithmicendparfor}

\algblock{MyFunc}{EndMyFunc}
\algrenewtext{MyFunc}[3]{#1\ \textbf{#2}\ #3}
\algrenewtext{EndMyFunc}{}

\algrenewcommand\algorithmicindent{1.0em}%
\newcommand{\MyComment}[1]{}
\algdef{SE}[DOWHILE]{Do}{DoWhile}{\algorithmicdo}[1]{\algorithmicwhile\ #1}

\algblock{MyWhile}{EndMyWhile}
\algnewcommand\algorithmicmywhile{\textbf{while}}
\algnewcommand\algorithmicmywhiledo{\textbf{do}}
\algnewcommand\algorithmicendmywhile{\textbf{end while}}
\algrenewtext{MyWhile}[1]{\algorithmicmywhile\ #1\ \algorithmicpardo}
\algrenewtext{EndMyWhile}{\algorithmicendmywhile}

\theoremstyle{plain}

\theoremstyle{definition}

\newcommand{\sysname}{\textsc{Peregrine}}

\newcommand\blindfootnote[1]{%
  \begingroup
  \renewcommand\thefootnote{}\footnote{#1}%
  \addtocounter{footnote}{-1}%
  \endgroup
}

\begin{document}

\title{\sysname{}: A Pattern-Aware Graph Mining System~\footnotemark}
\renewcommand{\shorttitle}{\sysname{}: A Pattern-Aware Graph Mining System}

\author{Kasra Jamshidi}
\affiliation{%
\department{School of Computing Science}
  \institution{Simon Fraser University}
  \state{British Columbia, Canada}
}
\email{kjamshid@cs.sfu.ca}

\author{Rakesh Mahadasa}
\affiliation{%
\department{School of Computing Science}
  \institution{Simon Fraser University}
  \state{British Columbia, Canada}
}
\email{rmahadas@cs.sfu.ca}

\author{Keval Vora}
\affiliation{%
\department{School of Computing Science}
  \institution{Simon Fraser University}
  \state{British Columbia, Canada}
}
\email{keval@cs.sfu.ca}

\begin{abstract}
Graph mining workloads aim to extract structural properties of a graph by exploring its subgraph structures. General purpose graph mining systems provide a generic runtime to explore subgraph structures of interest with the help of user-defined functions that guide the overall exploration process. However, the state-of-the-art graph mining systems remain largely oblivious to the shape (or pattern) of the subgraphs that they mine. This causes them to: (a) explore unnecessary subgraphs; (b) perform expensive computations on the explored subgraphs; and, (c) hold intermediate partial subgraphs in memory; all of which affect their overall performance. Furthermore, their programming models are often tied to their underlying exploration strategies, which makes it difficult for domain users to express complex mining tasks.

In this paper, we develop \mbox{\sysname{}}, a pattern-aware graph mining system that directly explores the subgraphs of interest while avoiding exploration of unnecessary subgraphs, and simultaneously bypassing expensive computations throughout the mining process. We design a pattern-based programming model that treats \emph{graph patterns} as first class constructs and enables \sysname{} to extract the semantics of patterns, which it uses to guide its exploration. 
Our evaluation shows that \sysname{} outperforms state-of-the-art distributed and single machine graph mining systems, and scales to complex mining tasks on larger graphs, while retaining simplicity and expressivity with its `pattern-first' programming approach.
\end{abstract}

\maketitle

\blindfootnote{* This is the full version of the paper appearing in the European Conference on Computer Systems (EuroSys), 2020.}

\section{Introduction}
Graph mining based analytics has become popular across various important domains including bioinformatics, computer vision, and social network analysis~\cite{motifsbiology,fsmcompvis,cliquescompchem,bearman-chains,sna-police,fakenews}.
These tasks mainly involve computing structural properties of the graph, i.e., exploring and understanding the substructures within the graph. Since the search space is exponential, graph mining problems are computationally intensive and their solutions are often difficult to program in a parallel or distributed setting. 

To address these challenges, general-purpose graph mining systems like Arabesque~\cite{arabesque}, RStream~\cite{rstream}, Fractal~\cite{fractal}, G-Miner~\cite{gminer} and AutoMine~\cite{automine} provide a generalized exploration framework and allow user programs to guide the overall exploration process. 
At the heart of these graph mining systems is an exploration engine that exhaustively searches subgraphs of the graph, and a series of filters that prune the search space to continue exploration for only those subgraphs that are of interest (e.g., ones that match a specific pattern) and are unique (to avoid redundancies coming from structural symmetries). The exploration happens in a step-by-step fashion where small subgraphs are iteratively extended based on their connections in the graph. As these subgraphs are explored, they get verified via \emph{canonicality checks} to guarantee uniqueness, and get analyzed via \emph{isomorphism computations} to understand their structure (or pattern). After that, the subgraphs either get pruned out because they don't match the pattern of interest, or are forwarded down the pipeline where their information is aggregated at the pattern level.

\begin{table}
\small
\setlength{\tabcolsep}{1mm}
  \begin{tabular}{c | c c c c c}
   & Arabesque & Fractal  & G-Miner & RStream & PRG-U \\ \midrule
   \sysname{} & 2-1317$\times$ & 1.1-737$\times$ & 3-131$\times$ & 2-2016$\times$ & 2-42$\times$ \\ \midrule
 \end{tabular}
  \vspace{0.05in}
  \caption{\sysname{} performance summary. PRG-U indicates \sysname{} without symmetry breaking, to model systems that are not fully pattern-aware (e.g., AutoMine).}
\label{tab-summaryresults}
\vspace{-0.2in}
\end{table}

While such an exploration process is general enough to compute different mining use cases including Frequent Subgraph Mining and Motif Counting, we observe that it remains largely oblivious to the patterns that are being mined. Hence, state-of-the-art graph mining systems 
face three main issues, as described next:
(1) These systems perform a large number of unnecessary computations; specifically, every subgraph explored from the graph, even in intermediate steps, is processed to ensure canonicality, and is analyzed to either extract its pattern or to verify whether it is isomorphic to another pattern. Since the exploration space for graph mining use cases is very large, performing those computations on every explored subgraph severely limits the performance of these systems. 
(2) The exhaustive exploration in these systems ends up generating a large amount of intermediate subgraphs that need to be held (either in memory or on disk) so that they can be extended. While systems based on breadth-first exploration~\cite{arabesque,rstream} demand high memory capacity, other systems like Fractal~\cite{fractal} and AutoMine~\cite{automine} use guided exploration strategies to reduce this impact; however, because they are not fully pattern-aware, they process a large number of intermediate subgraphs which severely limits their scalability as graphs grow large.
(3) The programming model in these systems is strongly tied to the underlying exploration strategy, which makes it difficult for domain experts to express complex mining use cases. For example, subgraphs containing certain pairs of strictly disconnected vertices (i.e., absence of edges) are useful for providing recommendations based on missing edges; mining such subgraphs with constraints on their substructure cannot be directly expressed in any of the existing systems.

In this paper, we take a `pattern-first' approach towards building an efficient graph mining system. We develop \mbox{\sysname{}}~\footnote{\sysname{} source code: https://github.com/pdclab/peregrine}, a pattern-aware graph mining system that directly explores the subgraphs of interest while avoiding exploration of unnecessary subgraphs, and simultaneously bypassing expensive computations (isomorphism and canonicality checks) throughout the mining process. \sysname{} incorporates a pattern-based programming model that enables easier expression of complex graph mining use cases, and reveals patterns of interest to the underlying system. Using the pattern information, \sysname{} efficiently mines relevant subgraphs by performing two key steps. First, it analyzes the patterns to be mined in order to understand their substructures and to generate an exploration plan describing how to efficiently find those patterns. And then, it explores the data graph using the exploration plan to guide its search and extract the subgraphs back to the user space. 

Our pattern-based programming model treats \emph{graph} \mbox{\emph{patterns}} as first class constructs: it provides basic mechanisms to load, generate and modify patterns along with interfaces to query patterns in the data graph. Furthermore, we introduce two novel abstractions, an \textsc{Anti-Edge} and an \textsc{Anti-Vertex}, that express advanced structural constraints on patterns to be matched. This allows users to directly operate on patterns and express their analysis as `pattern programs' on \sysname{}. Moreover, it enables \sysname{} to extract the semantics of patterns which it uses to generate efficient exploration plans for its pattern-aware processing model. 

We rely on theoretical foundations from existing subgraph matching research~\cite{po,postponecart} to generate our exploration plans. 
Since \sysname{} directly finds the subgraphs of interest, it does not incur additional processing over those subgraphs throughout its exploration process; this directly results in much lesser computation compared to the state-of-the-art graph mining systems. Moreover, \sysname{} does not maintain intermediate partial subgraphs in memory, resulting in much lesser memory consumption compared to other systems.

\sysname{} runs on a single machine and is highly concurrent. We demonstrate the efficacy of \sysname{} by evaluating it on several graph mining use cases including frequent subgraph mining, motif counting, clique finding, pattern matching (with and without structural constraints), and existence queries. Our evaluation on real-world graphs shows that \sysname{} running on a single 16-core machine outperforms state-of-the-art distributed graph mining systems including Arabesque~\cite{arabesque}, Fractal~\cite{fractal} and G-Miner~\cite{gminer} running on a cluster with eight 16-core machines; and significantly outperforms RStream~\cite{rstream} running on the same machine. Furthermore, \sysname{} could easily scale to large graphs and complex mining tasks which could not be handled by other systems. \rtab{tab-summaryresults} summarizes \sysname{}'s performance.

\section{Background \& Motivation}
We first briefly review graph mining fundamentals, and then discuss performance and programmability issues in state-of-the-art graph mining systems. In the end, we give an overview of \sysname{}'s pattern-aware mining techniques.

\begin{figure*}[t]
   \centering
   \begin{minipage}{0.49\textwidth}
     \subfloat[Step-by-step exploration in graph mining systems starting at vertex 1 and vertex 3. In total, 13 partial matches get explored and 13 canonicality checks are performed that prune out 5 partial matches. Isomorphism checks are performed on the remaining 8 matches for applications like FSM.]{
    \includegraphics[height=1.6in]{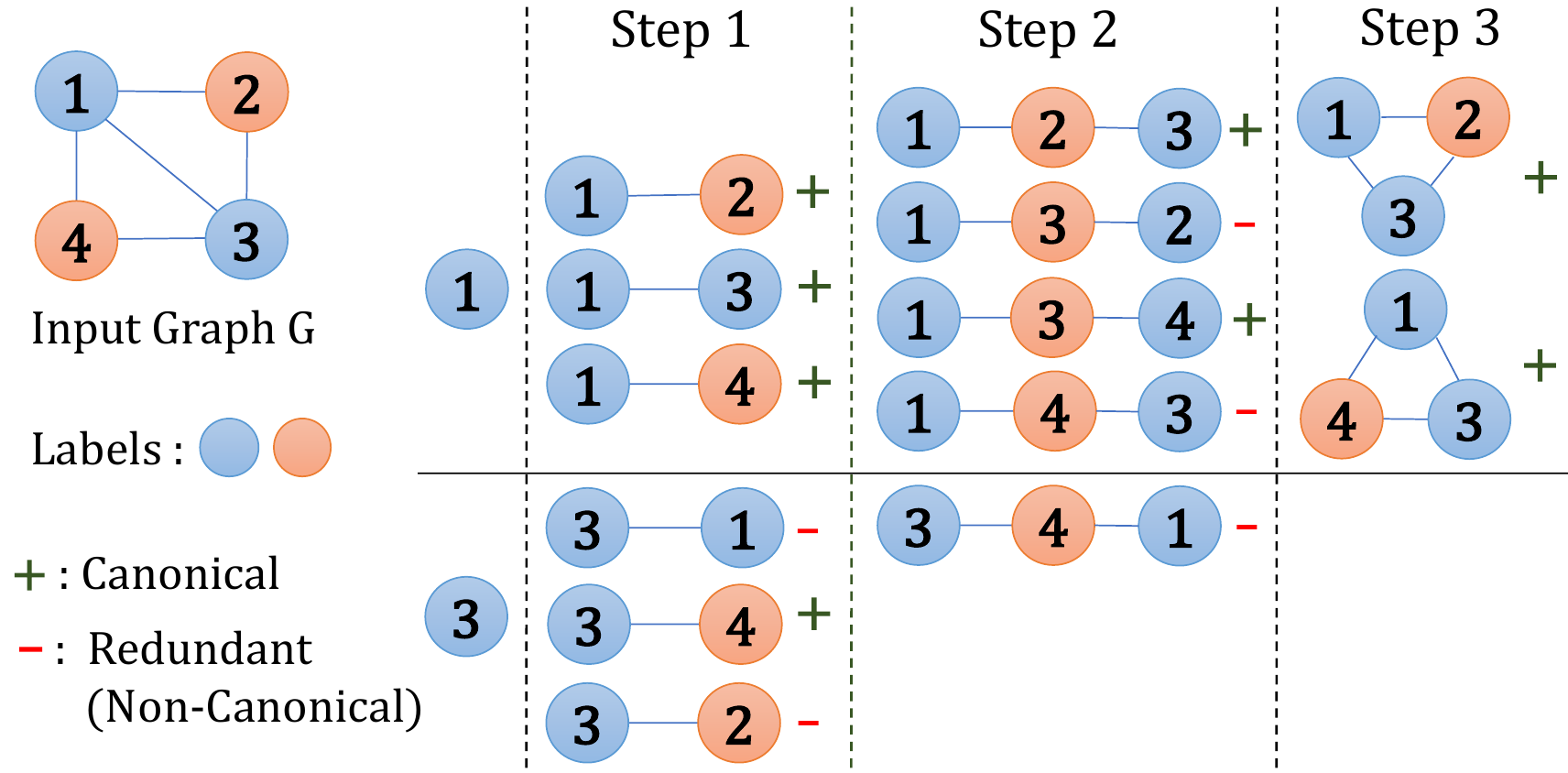}
    \label{fig-unawarematchingexample}
    }
    \end{minipage}
    \begin{minipage}{0.49\textwidth}
    {
    \small
    \subfloat[Profiling results for 4-Clique Counting on Patents~\cite{patents} which contains $\sim$3.5M cliques of size 4. Isomorphism counts are 0 for RStream and Fractal because they have native support for clique computation.]{
        \begin{tabular}{c|ccc}
         \multirow{2}{*}{System} & Total & Canonicality & Isomorphism \\
         & Matches &  Computations & Computations \\  \hline
          RStream & 1.2B (342$\times$) & 33.0M & 0  \\ 
          Arabesque & 1.4B (400$\times$) & 1.4B & 3.5M \\ 
          Fractal & 659.0M (188$\times$) & 599.6M & 0 \\ 
        \hline
        \end{tabular}
        \label{tab-unawarecliquenumbers}
    }
    
    \subfloat[Profiling results for 3-Motif Counting on Patents~\cite{patents} which contains $\sim$320M 3-sized motifs.]{
        \begin{tabular}{c|ccc}
         \multirow{2}{*}{System} & Total & Canonicality & Isomorphism \\
         & Matches &  Computations & Computations \\  \hline
          RStream & 40.1B (125$\times$) & 40.1B & 343.3M \\ 
          Arabesque & 685.8M (2.1$\times$) & 685.8M & 320.7M \\ 
          Fractal  & 665.6M (2.1$\times$) & 649.1M & 320.7M \\ 
         \hline
        \end{tabular}
        \label{tab-unawaremotifnumbers}
    }
    }
    \end{minipage}
    \vspace{-0.1in}
    \caption{Left: Example illustrating step-by-step exploration; Right: Number of matches explored (partial and full), canonicality checks performed, and isomorphism checks performed by RStream~\cite{rstream}, Arabesque~\cite{arabesque} and Fractal~\cite{fractal}. Numbers in brackets indicate the magnitude of matches explored relative to result size.}
    \label{fig-motivationplusdata}
    \vspace{-0.05in}
\end{figure*} 

\subsection{Graph Mining Overview}
\label{sec-graph-mining-overview}
\paragraph{Graph Terminology.} 
Given a graph $g$, we use $V(g)$ and $E(g)$ to denote its set of vertices and edges respectively. If the graph is labeled, we use $L(g)$ to denote its set of labels. A \emph{subgraph} $s$ of $g$ is a graph containing a subset of edges in $g$ and their endpoints.
\paragraph{Graph Mining Model.} 
Graph mining problems involve finding subgraphs of interest in a given input graph. We use $P$ to denote the \emph{pattern graph} (representing structure of interest) and $G$ to denote the input \emph{data graph}. 
We define a \emph{\underline{match}} $M$ as a subgraph of $G$ that is \emph{isomorphic} to $P$, where isomorphism is defined as a one-to-one mapping between $V(P)$ and $V(M)$ such that if two vertices are adjacent in $P$, then their corresponding vertices are adjacent in $M$. There are two kinds of matches depending on how vertices and edges from $G$ are extracted in $M$. An \emph{edge-induced match} is any subgraph of $G$ that is isomorphic to $P$. A \emph{vertex-induced match} is a subgraph of $G$ that is isomorphic to $P$ while containing all the edges in $E(G)$ that are incident on $V(M)$.

Since there can be sub-structure symmetries within $P$ (e.g., a triangle structure looks the same when it is rotated), the same subgraph of $G$ can result in multiple different matches each with a different one-to-one mapping with $V(P)$. These matches are \emph{automorphisms} of each other, where automorphic matches are defined as two matches $M_1$ and $M_2$ such that $V(M_1) = V(M_2)$. A \emph{canonical match} is the unique representative of a set of automorphic matches. Hence, uniqueness is ensured by choosing the canonical match from every set of automorphic matches in $G$.

\newpage

While the techniques presented in this paper work for both directed and undirected graphs, for easier exposition, we assume that $P$ and $G$ are undirected.

\vspace{-0.1in}
\paragraph{Graph Mining Problems.}
We briefly describe common graph mining problems. While the problems listed below focus on \emph{counting} subgraphs of interest, they are often generalized to \emph{listing} (or \emph{enumerating}) as well. 
\\[0.03in] 
\emph{--- Motif Counting.} 
A motif is any connected, unlabeled graph pattern. The problem involves counting the occurrences of all motifs in $G$ up to a certain size. 
\\[0.03in] 
\emph{--- Frequent Subgraph Mining (FSM).} 
The problem involves listing all labeled patterns with $k$ edges that are frequent in $G$ (i.e., frequency of their matches in $G$ exceed a threshold $\tau$). The frequency of a pattern (also called support) is measured in a variety of ways~\cite{mis-support,app-specific-support,independence-support,merged-support}, but most systems choose the \emph{minimum node image (MNI)}~\cite{fsm-mni} support measure since it can be computed efficiently. MNI is anti-monotonic, i.e., given two patterns $p$ and $p'$ such that $p$ is a subgraph of $p'$, support of $p$ will be at least as high as that of $p'$. 
\\[0.03in] 
\emph{--- Clique Counting.}
A $k$-clique is a fully-connected graph with $k$ vertices. The problem involves counting the number of $k$-cliques in $G$. Variations of this problem include counting \emph{pseudo-cliques}, i.e., patterns whose edges exceed some density threshold; \emph{maximal cliques}, i.e., cliques that are not contained in any other clique; and, \emph{frequent cliques}, i.e., cliques that are frequent (exceeding a frequency threshold). \\[0.03in] 
\emph{--- Pattern Matching.}
The problem involves matching (counting) the number of subgraphs in $G$ that are isomorphic to a given pattern. A variation of this problem is counting \emph{constrained subgraphs}, i.e., subgraphs with structural constraints (e.g., certain vertices in the subgraph must not be adjacent). 

All of the above graph mining use cases can be modelled in 3 steps: pattern selection, pattern matching, and aggregation. 

\subsection{Issues with Graph Mining Systems}
\label{sec-issueswithgraphminingsystems}
While several graph mining systems have been developed \cite{arabesque,fractal,gminer,rstream,automine}, they are not pattern-aware. Hence, they demand high computation power and require large memory (or storage) capacity, while also lacking the ability to easily express mining programs at a high level. 

\subsubsection{Performance}
\paragraph{(A) High Computation Demand.} 
Graph mining systems explore subgraphs in a step-by-step fashion by starting with an edge and iteratively extending matches depending on the structure of the data graph. Since they do not analyze the structure of the pattern to guide their exploration, they perform a large number of: (a) unnecessary explorations; (b) canonicality checks; and, (c) isomorphism checks. 

\rfig{fig-unawarematchingexample} shows an example of step-by-step exploration starting from vertex 1 and vertex 3. In step 1, both the vertices get extended generating 6 partial matches each of size 1 (edges). These are tested for canonicality which prunes out $(3, 1)$ and $(3, 2)$ (non-canonical matches are marked with $-$). For applications like FSM, isomorphism checks are performed on each of the canonical matches to identify their structure and compute metrics. Then, the remaining 4 matches progress to the next step and the entire process repeats.
While explorations get pruned via both canonicality and isomorphism checks, every valid partial match is extended to multiple matches which may no longer be valid; generation of intermediate matches which do not result into valid final matches is unnecessary. Furthermore, all intermediate partial matches (unnecessary and valid matches) are operated upon to identify their structure (i.e., isomorphism check) and to verify their uniqueness (i.e., canonicality check). 
In our example, 13 intermediate matches get generated, 5 of which are unnecessary; 13 canonicality checks and 8 isomorphism checks are performed. 
If these checks are not performed at every step (as done in Fractal~\cite{fractal} by delaying its \texttt{filter} step), a massive amount of partial and complete matches that do not contribute to final result would get generated.

We verified the above behavior by profiling graph mining systems on clique counting and motif counting applications. As shown in \rfig{tab-unawarecliquenumbers} and \rfig{tab-unawaremotifnumbers}, on Patents~\cite{patents} (a real-world graph dataset), 
RStream~\cite{rstream} and Arabesque~\cite{arabesque} generate over a billion partial matches for clique counting while the total number of cliques is only $\sim$3.5M ($\sim$99.7\% matches were unnecessary); similarly for motif counting, RStream generates over 40 billion partial matches ($\sim$99.2\% unnecessary) and Arabesque generates over 685 million partial matches ($\sim$52\% unnecessary). They also perform a large number (hundreds of millions to billions) of canonicality checks and isomorphism checks. Since Fractal~\cite{fractal} explores in depth-first fashion, its numbers are better than RStream and Arabesque; however, they are still very high. 

\paragraph{(B) High Memory Demand.} 
Graph mining systems often hold massive amounts of (partial and complete) matches in memory and/or in external storage. 
Systems based on step-by-step exploration require valid partial matches so that they can be extended in subsequent steps; the total size (in bytes) required by all matches (partial and complete) quickly grows (often beyond main memory capacity) as the size of the pattern or data graph increases. Such a memory demand is lower in DFS-based exploration (as done in Fractal~\cite{fractal}). For clique and motif counting in \rfig{fig-motivationplusdata}, Arabesque consumes $\sim$101GB main memory while Fractal requires $\sim$32GB memory.
    
\subsubsection{Programmability}
Programming in graph mining systems is done at vertex and edge level, with semantics of constructing the required matches defined explicitly by user's mining program. 
This means, mining programs expressed in those systems contain the logic for: (a) validating partial and complete matches; (b) extending matches via edges and/or vertices; and, (c) processing the final valid matches. As the size of subgraph structure to be mined grows, the complexity of validating partial matches increases, making mining programs difficult to write. For example, the multiplicity algorithm to avoid over-counting in AutoMine~\cite{automine} cannot be used if the user wants to enumerate patterns, which leaves the responsibility of identifying unique matches to the user. Furthermore, complicated structural constraints beyond the presence of vertices, edges and labels cannot be easily expressed in any of the existing systems.

\subsection{Overview of \sysname{}}
We develop a pattern-aware graph mining system that directly finds subgraphs of interest without exploring unnecessary matches while simultaneously avoiding expensive isomorphism and canonicality checks throughout the mining process. We do so by designing a pattern-based programming model that treats \emph{graph patterns} as first class constructs, and by developing a processing model that uses the pattern's substructure to guide the exploration process.

\paragraph{Pattern-based Programming.}
In \sysname{}, graph mining tasks are directly expressed in terms of subgraph structures (i.e., graph patterns). Our pattern-aware programming model allows declaring (statically and dynamically generated) patterns, modifying patterns, and performing user-defined operations over matches explored by the runtime. This allows concisely expressing mining programs by abstracting out the underlying runtime details, and focusing only on the substructures to be explored. Moreover, we introduce two novel abstractions, \emph{anti-edges} and \emph{anti-vertices}: an anti-edge enforces strict disconnection between two vertices in the match whereas an anti-vertex captures strict absence of a common neighbor among vertices in the match. These abstractions allow users to easily express advanced structural constraints on patterns to be mined.

\paragraph{Automatic Generation of Exploration Plan.}
With patterns of interest directly expressed, \sysname{} analyzes the patterns and computes an \emph{exploration plan} which is later used to guide the exploration in the data graph. Specifically, the pattern is first analyzed to eliminate symmetries within itself so that expensive canonicality checks during exploration can be avoided. Then the pattern is reduced to its \emph{core substructure} that enables identifying matches using simple graph traversals and adjacency list intersection operations without performing explicit isomorphism checks.

\paragraph{Guided Pattern Exploration.}
After the exploration plan is generated, \sysname{} starts the exploration process using our pattern-aware processing model. The exploration process matches the core substructure of the pattern to generate partial matches using recursive graph traversals in the data graph. As partial matches get generated, they are extended to form final complete matches by intersecting the adjacency lists of vertices in the partial matches. Since the entire exploration is guided by the plan generated from the pattern of interest, the exploration does not require intermediate isomorphism and canonicality checks for any of the partial and complete matches that it generates. This reduces the amount of computation done in \sysname{} compared to state-of-the-art graph mining systems. Moreover, since matches are recursively explored and instantly extended to generate complete final results, partial state is not maintained in memory throughout the exploration process which significantly reduces the memory requirement for \sysname{}.

\vspace{0.1in}
Finally, we reduce load imbalance in \sysname{} by enforcing a strict matching order based on vertex degrees. 
Furthermore, we incorporate on-the-fly aggregation and early termination features to provide global updates as mining progresses so that exploration can be stopped once the conditions required to compute final results are met.

\begin{figure}[t]
\small
\begin{lstlisting}[numbers=none,xleftmargin=0.5em]
[L1] Set<Pattern> loadPatterns(String filename);
[G1] Set<Pattern> generateAllEdgeInduced(Int size);
[G2] Set<Pattern> generateAllVertexInduced(Int size);
[S1] Pattern generateClique(Int size);
[S2] Pattern generateStar(Int size);
[S3] Pattern generateChain(Int size);
[C1] Set<Pattern> extendByEdge(Set<Pattern> patterns);
[C2] Set<Pattern> extendByVertex(Set<Pattern> patterns);
\end{lstlisting}
\hrule
\begin{lstlisting}[numbers=none]
class Pattern {
	Set<Vertex> getNeighbors(Vertex u);
	Label getLabel(Vertex u);
	Bool areConnected(Vertex src, Vertex dst);
	Void addEdge(Vertex src, Vertex dst);
	Void addAntiEdge(Vertex src, Vertex dst);
	Void removeEdge(Vertex src, Vertex dst);
	Void addLabel(Vertex u, Label l);
	. . .
};
\end{lstlisting}
\vspace{-0.2in}
\caption{\sysname{} Pattern Interface.}
\label{fig-patternapi}
\vspace{-0.05in}
\end{figure}

\section{\sysname{} Programming Model}
Since graph mining fundamentally involves finding subgraphs that satisfy certain structural properties, we design our programming model around \emph{graph patterns} as first class constructs. This allows users to easily express the subgraph structures of interest, without worrying about the underlying mechanisms of how to explore the graph and find those structures. With such a declarative style of expressing patterns, \sysname{} enables users to program complex mining queries as operations over the matches. The clear separation of \emph{what to find} and \emph{what to do with the results} helps users to quickly reason about  correctness of their mining logic, and develop advanced mining-based analytics. 

We first present how patterns are directly expressed in \sysname{}, and then show how common graph mining use cases can be programmed with patterns in \sysname{}.

\subsection{\sysname{} Patterns}
\rfig{fig-patternapi} shows our API to directly express, construct and modify connected graph patterns. Patterns can be constructed statically and loaded using \texttt{[L1]}, or can be constructed dynamically \texttt{[G1-G2, C1-C2, S1-S3]}. \texttt{[G1]} and \texttt{[G2]} generate all unique patterns that can be induced by certain number of edges and vertices respectively. \texttt{[S1-S3]} generate special well-known patterns. \texttt{[C1-C2]} take a group of patterns as input, and extend one of them by an edge or a vertex, to return all of the unique new patterns that result from these extensions. This allows constructing patterns step-by-step which is useful to perform guided exploration. The \texttt{Pattern} class provides a standard interface to access and modify the pattern graph structure.

In most common applications, the edges and vertices in the pattern graph are sufficient for \sysname{} to find subgraph structures that match the pattern. For advanced mining use cases that require structural constraints within the pattern, we introduce \emph{anti-edges} and \emph{anti-vertices}. 

\subsubsection{Anti-Edges}
Anti-edges are used to model constraints between vertex pairs in the pattern. They are special edges indicating disconnections between pairs of vertices. 
For example, in a social network graph where vertices model people and edges model their friendships, extracting unrelated people with at least two mutual friends can be achieved using $p_a$ in \rfig{fig-not-examples} where $(u_2, u_4)$ is an anti-edge.

An anti-edge ensures that the two vertices in the match do not have an edge between them in the data graph. If two vertices $u_1$ and $u_2$ are connected via an anti-edge in a pattern $p$, then any match for $p$ guarantees the anti-edge constraint:
\begin{align*}
  m(u_1) = v_1 \land m(u_2) = v_2 \implies (v_1, v_2) \not\in E(G)
\end{align*}
where $m$ is the function mapping vertices in $p$ to vertices in $G$. The two vertices connected by an anti-edge are called \emph{anti-adjacent}. \rfig{fig-not-examples} shows another example pattern ($p_b$) with anti-adjacent vertices $u_1$-$u_3$ and $u_2$-$u_4$, and their corresponding anti-edges $(u_1, u_3)$ and $(u_2, u_4)$. \sysname{} natively supports anti-edges (discussed in \rsec{sec-matching-anti-edges}), and hence, it directly matches only those subgraphs that do \emph{not} contain an edge between the two anti-adjacent vertices.

\begin{figure}[t]
\centering
  \includegraphics[width=3.3in]{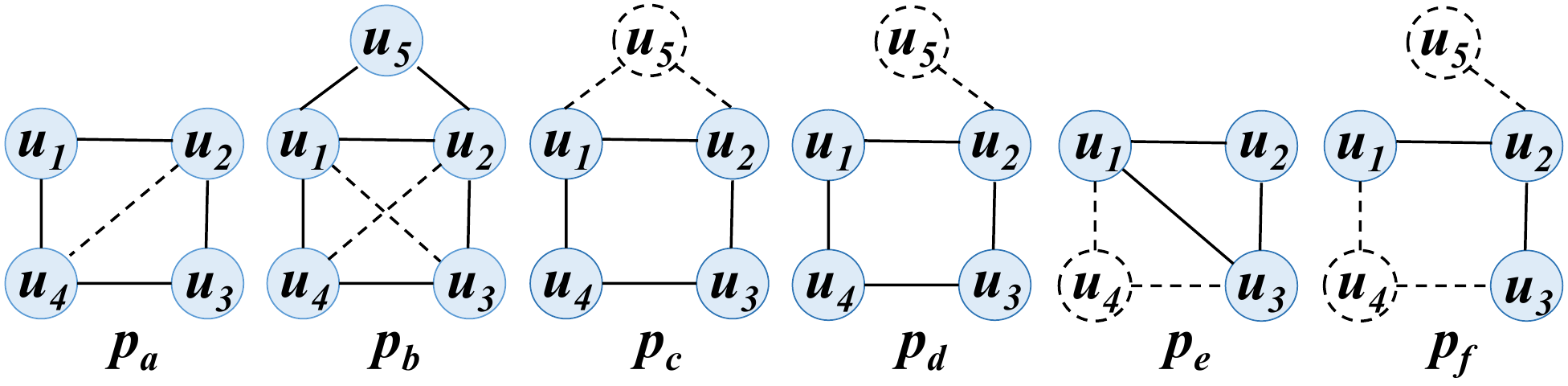}
  \vspace{-0.1in}
  \caption{Example patterns with Anti-Edges and Anti-Vertices.}
  \label{fig-not-examples}
  \vspace{-0.15in}
\end{figure}

\subsubsection{Anti-Vertices}
Anti-vertices are used to model constraints among shared neighborhoods of vertices in the pattern. They are special vertices that are only connected to other vertices via anti-edges. For example, extracting pairs of friends with only one mutual friend in a social network graph can be achieved using $p_e$ in \rfig{fig-not-examples} where $u_4$ is an anti-vertex. Such a query cannot be directly expressed using anti-edges alone. 

Anti-vertices represent absence of a vertex. 
So a match of a pattern with an anti-vertex will not contain a data vertex matching the anti-vertex. If $\overline{u}$ is an anti-vertex in a pattern $p$ and $S$ is the set of data vertices matching the neighbors of $\overline{u}$, then any match for $p$ guarantees the anti-vertex constraint:
\begin{align*}
  S = m(\adj{\overline{u}}) \implies \bigcap_{v \in S} \adj{v} \setminus m(\adj{m^{-1}(v)}) = \varnothing
\end{align*}
where $m$ is the function mapping vertices in $p$ to vertices in $G$, and $m^{-1}$ is its inverse.

\begin{figure*}
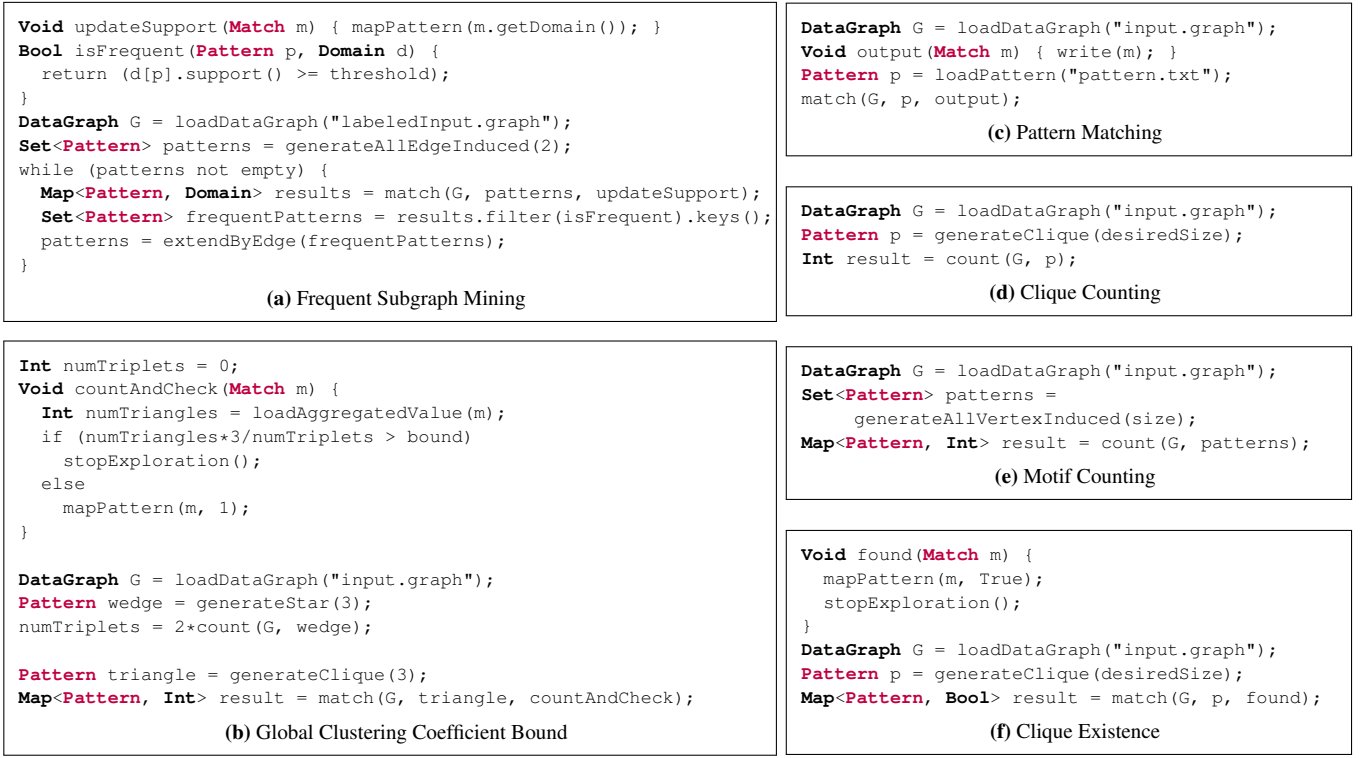

   \begin{minipage}{0.565\textwidth}
    {
    \begin{fminipage}{\textwidth}
      \subfloat[Frequent Subgraph Mining]{
        \makebox[\textwidth][l]{
\begin{lstlisting}[numbers=none]
Void updateSupport(Match m) { mapPattern(m.getDomain()); }
Bool isFrequent(Pattern p, Domain d) {
	return (d[p].support() >= threshold);
}
DataGraph G = loadDataGraph("labeledInput.graph");
Set<Pattern> patterns = generateAllEdgeInduced(2);
while (patterns not empty) {
  Map<Pattern, Domain> results = match(G, patterns, updateSupport);
  Set<Pattern> frequentPatterns = results.filter(isFrequent).keys();
  patterns = extendByEdge(frequentPatterns);
}
\end{lstlisting}}
        \label{fsm_example}
      }
      \end{fminipage}

      \begin{fminipage}{\textwidth}
      \subfloat[Global Clustering Coefficient Bound]{
        \makebox[\textwidth][l]{
\begin{lstlisting}[numbers=none]
Int numTriplets = 0;
Void countAndCheck(Match m) {
  Int numTriangles = loadAggregatedValue(m);
  if (numTriangles*3/numTriplets > bound)
    stopExploration();
  else
    mapPattern(m, 1);
}

DataGraph G = loadDataGraph("input.graph");
Pattern wedge = generateStar(3);
numTriplets = 2*count(G, wedge);

Pattern triangle = generateClique(3);
Map<Pattern, Int> result = match(G, triangle, countAndCheck);
\end{lstlisting}}
        \label{clustering_example}
      }
       \end{fminipage}
    }
    \end{minipage}
    \hfill
    \begin{minipage}{0.41\textwidth}
    {
    \begin{fminipage}{\textwidth}
      \subfloat[Pattern Matching]{
        \makebox[\textwidth][l]{
\begin{lstlisting}[numbers=none]
DataGraph G = loadDataGraph("input.graph");
Void output(Match m) { write(m); }
Pattern p = loadPattern("pattern.txt");
match(G, p, output);
\end{lstlisting}}
        \label{pm_example}
      }
       \end{fminipage}

\vspace{0.06in}
    \begin{fminipage}{\textwidth}
      \subfloat[Clique Counting]{
        \makebox[\textwidth][l]{
\begin{lstlisting}[numbers=none]
DataGraph G = loadDataGraph("input.graph");
Pattern p = generateClique(desiredSize);
Int result = count(G, p);
\end{lstlisting}}
        \label{cliques_example}
      }
       \end{fminipage}

\vspace{0.06in}
    \begin{fminipage}{\textwidth}
      \subfloat[Motif Counting]{
        \makebox[\textwidth][l]{
\begin{lstlisting}[numbers=none]
DataGraph G = loadDataGraph("input.graph");
Set<Pattern> patterns = generateAllVertexInduced(size);
Map<Pattern, Int> result = count(G, patterns);
\end{lstlisting}}
        \label{motifs_example}
      }
       \end{fminipage}

\vspace{0.06in}             
     \begin{fminipage}{\textwidth}
      \subfloat[Clique Existence]{
        \makebox[\textwidth][l]{
\begin{lstlisting}[numbers=none]
Void found(Match m) {
  mapPattern(m, True);
  stopExploration();
}
DataGraph G = loadDataGraph("input.graph");
Pattern p = generateClique(desiredSize);
Map<Pattern, Bool> result = match(G, p, found);
\end{lstlisting}}
        \label{exists_example}
      }
       \end{fminipage}
  }
    \end{minipage}
    \vspace{-0.1in}
  \caption{Graph mining use cases in \sysname{}'s pattern-aware programming model.}
  \label{program_examples}
  \vspace{-0.1in}
\end{figure*}

To distinguish anti-vertices, we call a vertex with at least one regular edge (i.e., not anti-edge) a regular vertex. 
\rfig{fig-not-examples} shows different patterns with anti-vertices:
In a match for $p_c$, the matches for $u_1$ and $u_2$ have no common neighbors. On the other hand, in a match for $p_d$, the match for $u_2$ has no neighbors other than the matches for $u_1$ and $u_3$. 
Finally, $p_f$ has two anti-vertices which combines constraints from $p_c$ and $p_d$.
\sysname{} natively supports anti-vertices (discussed in \rsec{sec-matching-anti-vertices}), and hence, it directly matches only those subgraphs that satisfy absence of vertices across neighborhoods as defined by anti-vertex constraint.

\subsubsection{Edge-Induced and Vertex-Induced Patterns}
Depending on the mining use case, the matches for a given pattern must be either \emph{edge-induced} or \emph{vertex-induced}. For example, Frequent Subgraph Mining (FSM) relies on edge-induced matches, whereas Motif Counting requires vertex-induced matches (programs shown in \rsec{sec-pattern-aware-mining-programs}). Our pattern-based programming model allows exploring subgraphs in both edge-based and vertex-based fashion as discussed next. 

An edge-induced match is a subgraph $s_e$ of $G$ such that the subgraph of $G$ induced by $E(s_e)$ is isomorphic to $p$. Note that, by definition, 
the subgraph induced by $E(s_e)$ is equal to $s_e$. Hence, edge-induced matches are directly expressed by the pattern.

A \emph{vertex-induced} match, on the other hand, is a subgraph $s_v$ of $G$ such that the subgraph of $G$ induced by $V(s_v)$ is isomorphic to $p$. Hence, the sets of vertex-induced and edge-induced matches of $p$ are not equal mainly because the vertices of $s_e$ can have more edges between them in $G$ than are present in $p$ (in general, the subgraph induced by $V(s_e)$ is not isomorphic to $p$). In our pattern-based programming approach, the vertex-induced requirement gets directly expressed using anti-edges. Specifically, to find vertex-induced matches of a pattern $p$, we use the following result:
\begin{theorem}
Let $p$ be a pattern, and $p'$ be another pattern such that $p'$ has the same vertices and edges as $p$, and every pair of vertices in $p$ that are not adjacent are anti-adjacent in $p'$. The set of vertex-induced matches of $p$ is equal to the set of edge-induced matches of $p'$.
\end{theorem}
\begin{proof}
  To prove equivalence of the two sets, we first show that every edge-induced match of $p'$ is a vertex-induced match of $p$, and then we show that every vertex-induced match of $p$ is an edge-induced match of $p'$.
 
  Let $m$ be an edge-induced match for $p'$.
  Observe that as $p$ and $p'$ contain the same edges and vertices, $m$ is also a match for $p$.
  By definition of $p'$, for any edge $(u, v)\not\in E(p)$, there is an anti-edge constraint between $u$ and $v$ in $p'$.
  Since $m$ is a match for $p'$, it satisfies this anti-edge constraint, such that $m(u)$ and $m(v)$ are not adjacent in the data graph.
  This means there is an edge between $m(u)$ and $m(v)$ if and only if $(u, v) \in E(p)$.
  Therefore, $m$ is a vertex-induced match of $p$.

  Conversely, let $m$ be a vertex-induced match for $p$.
  Since $m$ is isomorphic to $p$, it contains all edges of $p$.
  Furthermore, $m$ is vertex-induced, so if a pair of pattern vertices $u_1, u_2$ in $p$ are not adjacent, then the corresponding data vertices $m(u_1)$ and $m(u_2)$ are not adjacent either.
  Hence, $m$ satisfies the anti-edge constraint for $u_1, u_2$.
  As this holds for all pairs of non-adjacent vertices in $p$, $m$ is also a match for $p'$.
 
\end{proof}

\vspace{-0.03in}
Hence, our pattern-based programming doesn't need to separately define the exploration strategy, as done in other pattern-unaware systems~\cite{arabesque,rstream}. 

\subsection{Pattern-Aware Mining Programs in \sysname{}}
\label{sec-pattern-aware-mining-programs}
\rfig{program_examples} shows \sysname{} programs for motif counting, frequent subgraph mining (FSM), clique counting, pattern matching, an existence query for global clustering coefficient bound, and an existence query for k-sized clique. All the programs first express patterns by dynamically generating them or by loading them from external source. Then they invoke \sysname{} engine to find (\texttt{match()}) and process matches of those patterns. For every match for the pattern, user-defined function (e.g., \texttt{updateSupport()}, \texttt{countAndCheck()}, \texttt{found()}, etc.) gets invoked to perform desired analysis. The \texttt{count()} function is a syntactic sugar and is equivalent to \texttt{match()} with a function that increments a counter. Most of the programs are straightforward; we discuss FSM and existence queries in more detail.

\subsubsection{FSM: Anti-Monotonicity \& Label Discovery}
FSM leverages anti-monotonicity in support measures (discussed in ~\rsec{sec-graph-mining-overview}). \sysname{} natively provides MNI support computation where it internally constructs the \emph{domain} of patterns, i.e., a table mapping vertices in $G$ to those in $p$ (similar to \cite{scalemine}). After exploration ends for a single iteration, the support measure maintained by \sysname{} can be directly used to prune infrequent patterns using a threshold, as shown in \rfig{fsm_example}, and only the remaining frequent patterns are then programmatically extended to be explored. 

Before finding the first small frequent labeled patterns, the FSM program has no information about which labelings are frequent. \sysname{} provides dynamic label discovery by starting with unlabeled (or partially labeled) patterns as input and returning labeled matches. Hence, the FSM program in \rfig{fsm_example} first starts with unlabeled patterns of size 2, and discovers frequent labeled patterns. 
It then iteratively extends the frequent labeled patterns with unlabeled vertices to discover frequent labeled patterns of larger sizes.

\begin{figure}[t]
  \begin{lstlisting}[mathescape=True,numbers=none,basicstyle=\ttfamily\footnotesize]
ExplorationPlan generatePlan(Pattern p) {
  partialOrders = breakSymmetries(p);
  vc = minConnectedVertexCover(p);
  p$_C$ = vertexInducedSubgraph(vc, p);
  matchingOrders = computeMatchingOrders(p$_C$, partialOrders);
  return (p$_C$, partialOrders, matchingOrders);
}
  \end{lstlisting}
  \vspace{-0.15in}
  \caption{Computing exploration plan.}
  \label{alg-queryeval}
  \vspace{-0.1in}
\end{figure}

\subsubsection{Existence Queries}
\label{sec-existence-queries}
Existence queries allow quickly verifying whether certain structural properties hold within a given data graph. \sysname{} allows dynamically stopping exploration when the required conditions get satisfied.

\rfig{clustering_example} shows a \sysname{} program to verify if the global clustering coefficient~\cite{globalclustering} of graph $G$ is above a certain bound. The global clustering coefficient is the ratio of three times the number of triangles and the number of triplets (all connected subgraphs with three vertices, including duplicates) in $G$. The number of triplets is equal to twice the number of edge-induced 3-star matches since the endpoints of a 3-star are symmetric. Hence, the program quickly computes the number of 3-stars, and then starts counting triangles. During exploration, if the number of triangles reaches the requisite number to exceed the bound, exploration stops immediately. 

\rfig{exists_example} shows \sysname{} program to check whether a clique of a certain size is present in $G$. As soon as the exploration finds at least one match, it stops and returns \texttt{True}.

\section{Pattern-Aware Matching Engine}
\label{sec-pattern-aware-matching-engine}
\sysname{} is pattern-aware, and hence, it directly finds patterns in any given data graph. In this section, we discuss our core pattern matching engine that directly finds canonical subgraphs from a given vertex in the data graph. In \rsec{sec-sysname-patten-aware-mining}, we will use this engine to build \sysname{}. For simplicity, we assume the data graph and the pattern are unlabeled.

\subsection{Directly Matching A Given Pattern}
\label{sec-directly-matching}
To avoid the overheads of a straightforward exhaustive search, 
we develop our pattern matching solution based on well-established techniques~\cite{po,postponecart,dualsim}. Since patterns are much smaller than the data graph, we analyze the given pattern to develop an exploration plan. This plan guides the data graph exploration to ensure generated matches are unique.

\begin{figure}[t]
  \includegraphics[height=1.1in]{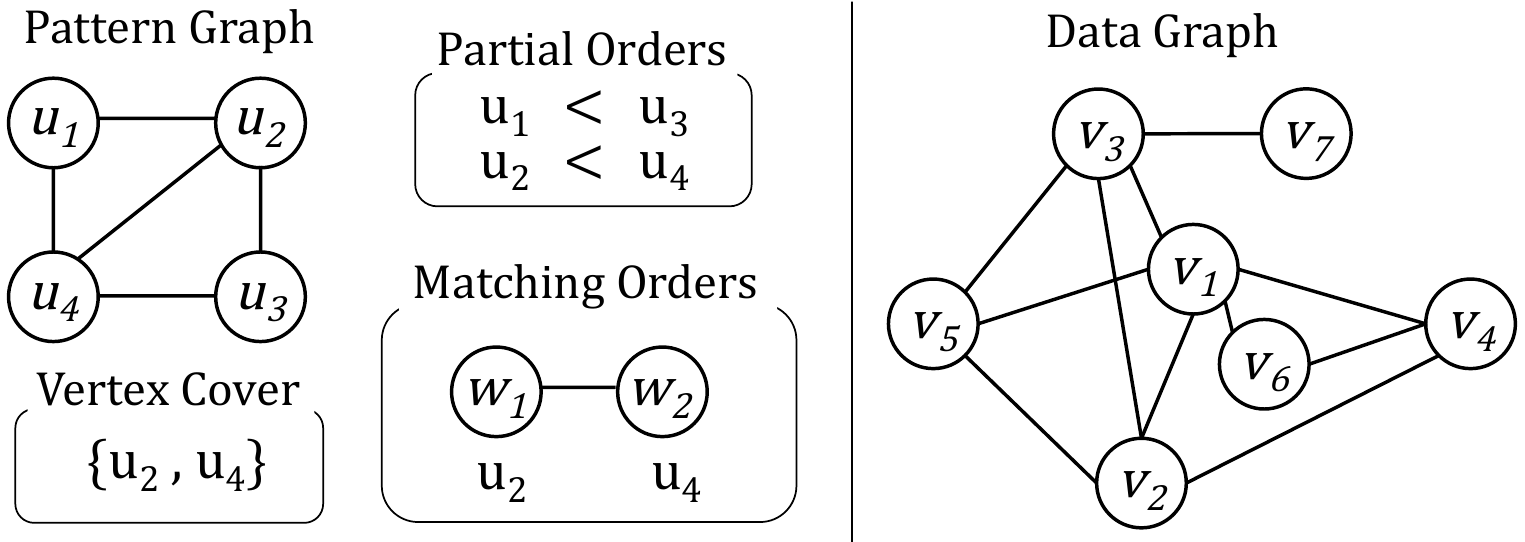}
  \vspace{-0.05in}
  \caption{Example of a pattern graph and a data graph.}
  \vspace{-0.1in}
  \label{fig-matching-example-1}
\end{figure}

\rfig{alg-queryeval} shows how the exploration plan is computed from a given pattern $p$. First, to avoid non-canonical matches we break the symmetries of $p$ by enforcing a partial ordering on matched vertices~\cite{po}.
This involves enumerating all automorphisms of $p$ to identify symmetries, and iteratively ordering pairs of symmetric vertices until the only automorphism satisfying the ordering is the one mapping each vertex to itself. For our example pattern in \rfig{fig-matching-example-1}, we obtain the partial ordering $u_1 < u_3$ and $u_2 < u_4$. 

In the next step, we compute the core of $p$ (called $p_C$) as the subgraph induced by its minimum connected vertex cover~\footnote{A connected vertex cover is a subset of connected vertices that covers all edges.}. Given a match $m$ for $p_C$, all matches of $p$ which contain $m$ can be computed from the adjacency lists of vertices in $m$.  
In our example, $p_C$ is the subgraph induced by $u_2$ and $u_4$.

To simplify the problem of matching $p_C$, we generate matching orders to direct our exploration in the data graph.
A matching order is a graph representing an ordered view of $p_C$. The vertices of the matching order are totally-ordered such that the partial ordering of $V(p)$ restricted to $V(p_C)$ is maintained. This allows matching $p_C$ by traversing vertices with increasing vertex ids without canonicality checks.

\begin{figure}[t]
\begin{lstlisting}[mathescape=true,numbers=none,escapeinside={{*}{*}},basicstyle=\ttfamily\footnotesize,xleftmargin=0em]
Void matchFrom(Match m, Pattern p, Func f, MatchingOrder mo, PartialOrder po, Int i) {
  if (i > $|V(p_C)|$) {
    // remaps m as in *\rsec{sec-directly-matching}*, before completing it
    // and invoking the user's callback f()
    completeMatch(m, p, f, po, 1);
  } else {
    for (v in getExtensionCandidates(mo, po, i)) {
      matchFrom(m+v, p, f, mo, po, i+1);
    }
  }
}

AggregationVal match(Graph G, Pattern p, Func f) {
  Aggregator a;
  (p$_C$,  partialOrder, matchingOrders) = generatePlan(p);
  parallel for (v in G) {
    for (matchingOrder in matchingOrders) {
      matchFrom({v}, p, f, matchingOrder,
                partialOrders, 1);
    }
  }
  return a.result();
}
\end{lstlisting}
  \vspace{-0.15in}
  \caption{Pattern-Aware Processing Model.}
  \label{alg-processingmodel}
  \vspace{-0.2in}
\end{figure}

We compute matching orders by enumerating all sequences of vertices in $p_C$ that meet the partial ordering, and for each sequence we create a copy of $p_C$ where the id of each vertex is remapped to its position in the sequence. Then, we discard duplicate matching orders. For our example pattern (\rfig{fig-matching-example-1}), its core substructure has only one valid vertex sequence, $\{u_2, u_4\}$, so we obtain only one matching order. Note that there can be multiple matching orders for a given $p_C$ depending on the partial orders. We call the $i^{\text{th}}$ matching order $p_{Mi}$.

Thus, to match $p_C$ it suffices to match its matching orders $p_{Mi}$. A match for $p_{Mi}$ results in 1 match for $p_C$ per valid vertex sequence. In our example, a match for $p_{M1}$, say $\{v_2, v_3\}$, is converted to a single match for $p_C$, $v_2 \to w_1 \to u_1, v_3 \to w_2 \to u_2$.

It is important to note that the exploration plan is generated by analyzing the pattern graph only, i.e., all the computations explained above are applied on $p$ (and its derivatives). Hence, exploration plans are computed quickly (often in less than half a millisecond).

\vspace{-0.05in}
\subsection{Matching Anti-Edges}
\label{sec-matching-anti-edges}
To enforce an anti-edge constraint, we perform a set difference between the adjacency lists of its endpoints. For example, if $v_i, v_j$ match $u_1, u_2$ of $p_a$ in~\rfig{fig-not-examples}, the candidates for $u_4$ are the elements of $\adj{v_i} \setminus \adj{v_j}$.

To perform the set difference, we need to ensure that one of the vertices of the anti-edge is already matched so that its adjacency list is available. Hence, when computing the vertex cover we also cover the anti-edge by including one of its endpoints. When computing partial orders, however, we do not need to consider anti-edges since they don't generate automorphic matches and we only verify absence of edges (i.e., we never traverse through anti-edges).

\subsection{Matching Anti-Vertices}
\label{sec-matching-anti-vertices}
The anti-vertex semantics offer flexibility to match them differently compared to anti-edges. Anti-vertices break symmetries and do not impact the core graph (i.e., vertex cover computation) as described next.

\begin{figure}[t]
  \includegraphics[height=0.95in]{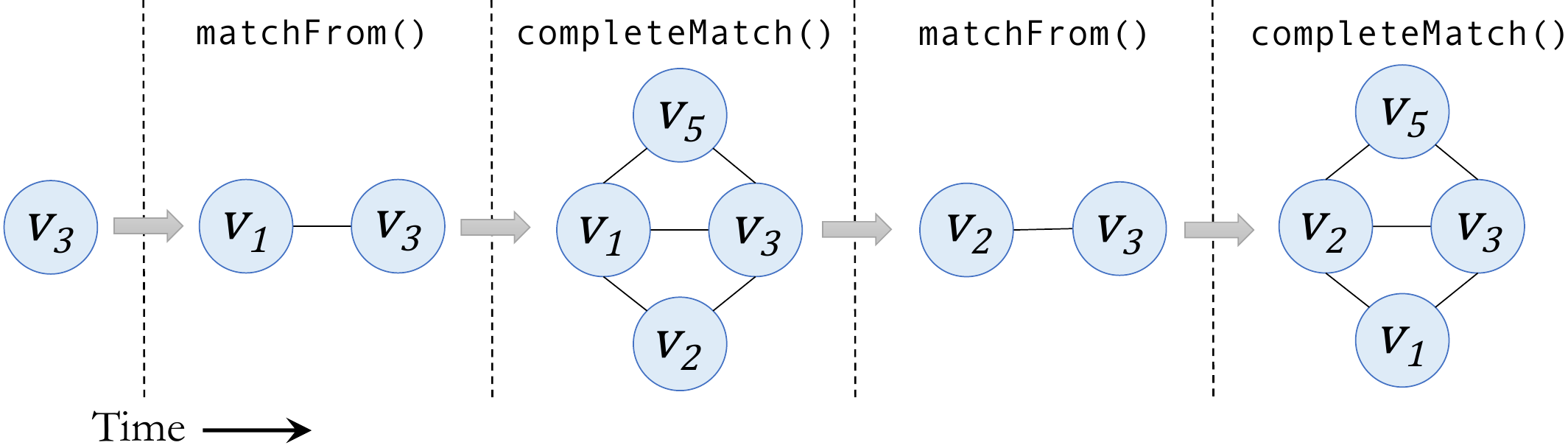}
  \caption{Pattern-guided exploration in \sysname{} for pattern and data graph in \rfig{fig-matching-example-1} with matching order high-to-low.}
  \label{fig-matching-example-2}
  \vspace{-0.1in}
\end{figure}

\paragraph{Checking Anti-Vertices.}
The anti-vertex constraint can only be verified after the common neighbors of an anti-vertex's neighbor have been matched. Thus, we perform the check after all true vertices are already matched. 

For example, consider $p_e$ in \rfig{fig-not-examples}, with anti-vertex $u_4$. If $v_i$, $v_j$, $v_k$ in the data graph match $u_1$, $u_2$, $u_3$ respectively, then we verify the anti-vertex constraint for $u_4$ as follows: \[(\adj{v_i}\setminus \{v_j\}) \cap (\adj{v_k}\setminus \{v_j\}) = \varnothing\]
Since anti-vertices do not participate while matching regular vertices and edges, we keep the pattern core $p_C$ the same as the core pattern when anti-vertices are removed.

\paragraph{Breaking Symmetries with Anti-Vertices.}
Anti-vertices introduce asymmetries in the pattern. We extract these asymmetries to ensure that valid matches do not get incorrectly pruned. We first explain how a partial ordering generated without considering anti-vertices would end up pruning valid matches, and then we showcase how anti-vertices get used to generate the correct partial ordering.

Let $p'$ be the sub-graph of $p_e$ (in \rfig{fig-not-examples}) with the anti-vertex removed. Notice that $p'$ is simply a triangle. In our data graph from \rfig{fig-matching-example-1}, the vertices $v_1, v_4, v_6$ form a triangle, and the pairs $\langle v_4, v_6\rangle$ and $\langle v_1, v_6\rangle$ have no common neighbors outside of the triangle. However, $\langle v_1, v_4\rangle$ have $v_2$ as a common neighbor. Thus, according to anti-vertex semantics, $p_e$ should match the subgraphs induced by $\{v_1, v_4, v_6\}$ and $\{v_1, v_6, v_4\}$. But since $p'$ does not include the anti-vertex $u_4$, it is fully symmetric and hence would only match $\{v_1, v_4, v_6\}$. Therefore, we cannot ignore anti-vertices in the symmetry-breaking algorithm: otherwise the resulting partial ordering would prune matches that are valid according to the anti-vertex constraint.

To generate a correct partial ordering, we expose the asymmetries introduced by anti-vertices to the symmetry-breaking algorithm, which treats the anti-edges of an anti-vertex differently than regular edges when computing automorphisms. Considering $p_e$ again, our symmetry-breaking algorithm finds that $u_2$ is not symmetric with $u_1$ and $u_3$ since it is not connected to the anti-vertex $u_4$. Meanwhile, $u_2$ and $u_4$ will not be considered symmetric either, because $u_2$ is connected to $u_1$ and $u_3$ with true edges whereas $u_4$ is connected with anti-edges.

\vspace{-0.05in}
\section{\sysname{}: Pattern-Aware Mining}
\label{sec-sysname-patten-aware-mining}
We will now discuss how \sysname{} performs pattern-aware mining using the matching engine presented in \rsec{sec-pattern-aware-matching-engine}. 

\vspace{-0.05in}
\subsection{Pattern-Aware Processing Model}
Mining in \sysname{} is achieved by matching patterns starting from each vertex and invoking the user function to process those patterns. Hence, a task in \sysname{} is defined as the data vertex where the matching process begins. As shown in \rfig{alg-processingmodel}, each mining task takes a start vertex and the exploration plan generated in \rsec{sec-pattern-aware-matching-engine} (matching orders, partial orders, pattern core $p_C$). From the starting vertex, we recursively match vertices in the matching order. At each recursion level, a data vertex is matched to a matching order vertex. To avoid non-canonical matches, we maintain sorted adjacency lists and use binary search to generate candidate sets comprised only of vertices that meet the total ordering.

Once a matching order is fully matched, it is converted to matches for $p_C$. Matches for $p_C$ are then completed by performing set intersections (for true edges) and set differences (for anti-edges) on sections of adjacency lists that satisfy the partial orders. Each completed match is passed to a user-defined callback for further processing. \rfig{fig-matching-example-2} shows a complete exploration example.

Note that our processing model doesn't incur expensive isomorphism and canonicality checks for every match in the data graph, while simultaneously avoiding mis-matches and only exploring subgraphs that match the given pattern. Furthermore, tasks in our processing model are independent of each other since explorations starting from two different vertices do not require any coordination. Threads dynamically pick up new tasks when they finish their current ones.

\subsection{Early Pruning for Dynamic Load Balancing} 
\label{sec-early-pruning-via-degree-aware-matching}
While a matching order enforces a total ordering on the data vertices matching $p_C$, there is flexibility in the order in which its vertices are matched.
To reduce the load imbalance across our matching tasks, we: (a) follow matching orders high-to-low, e.g. in our example in \rfig{fig-matching-example-1} we match $w_2$ before $w_1$; and, (b) order vertices by their degree such that $v_i < v_j$ in the data graph if and only if $degree(v_i) \leq degree(v_j)$.

High-degree vertices have fewer neighbors with degrees higher than or equal to their own, so the degree-based ordering ensures that when a high-degree vertex is matched to $w_2$, only those few neighbors can be matched to $w_1$. Thus, explorations of neighbors with lower degrees are pruned. Note that the total number of matches generated remains the same; the high-to-low matching order traversal, along with degree-based vertex ordering, reduces the workload imbalance of matching across high-degree and low-degree vertices by dynamically pruning more explorations from high-degree tasks while enabling those explorations in low-degree tasks.

Finally, it is important to note that this process does not `eliminate' workload imbalance simply because the mining workload is dynamic and depends on the pattern and data graphs. Hence, to avoid stragglers and maximize parallelism, we process tasks in the order defined by the degree of the starting vertex, beginning with the highest-degree vertices.

\vspace{-0.02in}
\subsection{Early Termination for Existence Queries}
For existence queries, \sysname{} allows actively monitoring the required conditions so that the exploration process terminates as quickly as possible. When the matching thread observes the required conditions, the user function calls \texttt{stopExploration()} to notify other matching threads. Threads monitor their notifications periodically while matching, and when a notification is observed, the thread-local values computed up to that point are aggregated and returned to the user.

\vspace{-0.02in}
\subsection{On-the-fly Aggregation}
\sysname{} performs on-the-fly aggregation to provide global updates as mining progresses. This is useful for early termination and for use cases like FSM where patterns that meet the support threshold can be deemed frequent while matching continues for other patterns. 

We achieve this using an asynchronous aggregator thread that periodically performs aggregation as values arrive from threads. The matching threads swap the global aggregation value with the local aggregation value and set a flag to indicate that new thread-local aggregation values are available for aggregation. The aggregator thread blocks until all thread-local aggregation values become available, after which it performs the aggregation and resets the flag to indicate that the global aggregation value is available. With this design, our matching threads remain non-blocking to retain high matching throughput.

\subsection{Implementation Details}
\sysname{} is implemented in C++ where concurrent threads operate on exploration tasks, each starting at a different vertex in the data graph. The data graph is represented using adjacency lists, and the tasks are distributed dynamically using a shared atomic counter indicating the next vertex to be processed. To minimize coordination, threads maintain information regarding their exploration tasks, including candidate sets for each pattern vertex as exploration proceeds.

\sysname{} provides native computation of support values for frequency-based mining tasks like FSM. Domains are implemented as a vector of bitmaps representing the data vertices that can be mapped to each pattern vertex. They are aggregated by merging their contents via logical-or. To scale to large datasets, we use compressed Roaring bitmaps~\cite{roaring}, which are more memory efficient than dense bitmaps.

\section{Evaluation}
We evaluate the performance of \sysname{} on a wide variety of graph mining applications and compare the results with the state-of-the-art general purpose graph mining systems~\footnote{We could not evaluate AutoMine~\cite{automine} since its source code is not available.}: Fractal~\cite{fractal}, Arabesque~\cite{arabesque}, RStream~\cite{rstream} and G-Miner~\cite{gminer}.

\begin{table}[t]
\small
  \begin{tabular}{l r r c r r}
    \multicolumn{1}{c}{\multirow{2}{*}{$G$}} & \multirow{2}{*}{$|V(G)|$} & \multirow{2}{*}{$|E(G)|$} & \multirow{2}{*}{$|L(G)|$} & Max. & Avg.  \\
     &  &  &  & Deg. & Deg. \\
    \midrule
(MI) Mico~\cite{mico}       & 100K    & 1M     & 29  & 96637 & 21.6 \\[0.05in]
(PA) Patents~\cite{patents} &         &        &     &       &      \\
    \hspace{6mm} $^{|}$\hspace{-0.5mm}--- \ Unlabeled & 3.7M    & 16M    & --- & 793   & 10   \\
    \hspace{6mm} $^{|}$\hspace{-0.5mm}--- \ Labeled   & 2.7M    & 13M    & 37  & 789   & 10   \\[0.05in]
(OK) Orkut~\cite{snap}      & 3M      & 117M   & --- & 33133 & 76   \\
(FR) Friendster~\cite{snap} & 65M     & 1.8B   & --- & 5214  & 55   \\
    \midrule
  \end{tabular}
  \vspace{0.1in}
  \caption{Real-world graphs used in evaluation. \\ '---' indicates unlabeled graph.}
  \label{datasets}
  \vspace{-0.3in}
\end{table}

\subsection{Experimental Setup}
\paragraph{System.}
All experiments were conducted on \texttt{c5.4xlarge} and \texttt{c5.metal} Amazon EC2 instances. Most experiments use \texttt{c5.4xlarge}, with an Intel Xeon Platinum 8124M CPU containing 8 physical cores (16 logical cores with hyper-threading), 32GB RAM, and 30GB SSD. Fractal (FCL), Arabesque (ABQ) and G-Miner (GM) were evaluated using both a cluster of 8 nodes (denoted by the suffix `-8'), as well as in single node configuration (denoted by the suffix `-1').

RStream was evaluated on a \texttt{c5.4xlarge} (RS-16) as well as a \texttt{c5.metal} (RS-96) equipped with an Intel Xeon Scalable Processor containing 48 physical cores (96 logical cores with hyper-threading), and 192GB RAM. Both instances were provisioned with a 500GB SSD.

In all performance comparisons, we ran \sysname{} on a \texttt{c5.4xlarge}, and we used \texttt{c5.metal} to study \sysname{}'s scalability and resource utilization.

\paragraph{Datasets.}
\rtab{datasets} lists the data graphs used in our evaluation. Mico (MI) is a co-authorship graph labeled with each author's research field. Patents (PA) is a patent citation graph. In the labeled version, each patent is labeled with the year it was granted. Orkut (OK) and Friendster (FR) are unlabeled social network graphs where edges represent friendships between users. Mico and labeled Patents have been used by previous systems~\cite{arabesque, fractal, rstream} to evaluate FSM while Orkut and Friendster were used by~\cite{gminer}. Except for FSM and labeled pattern matching, all experiments on Patents use its larger, unlabeled version.

\begin{figure}
  \includegraphics[height=0.67in]{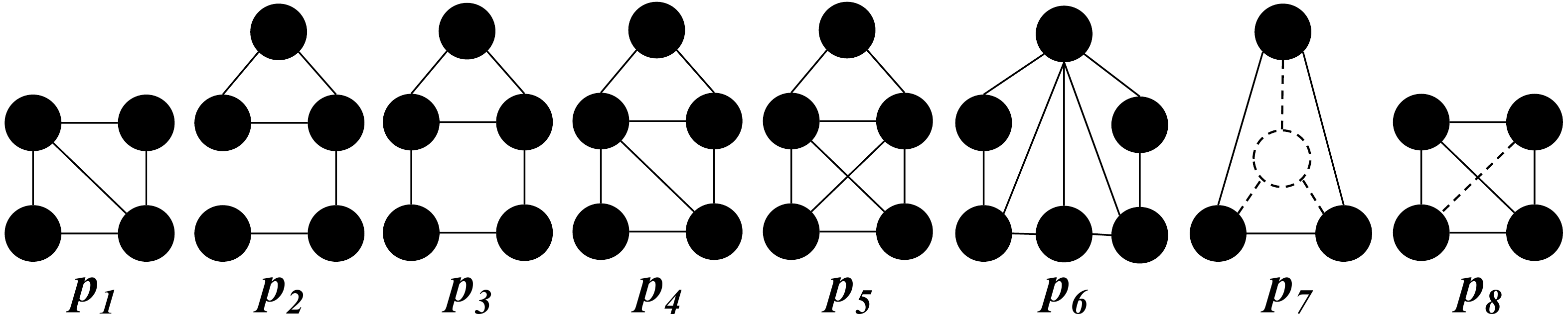}
  \caption{Patterns used in evaluation.}
  \label{fig-query-patterns}
  \vspace{-0.1in}
\end{figure}

\paragraph{Applications.}
We evaluated \sysname{} on a wide array of applications: counting motifs with 3 and 4 vertices, labeled 3- and 4-motifs; counting $k$-cliques, for $k$ ranging from 3 to 5; FSM with patterns of 3 edges on labeled datasets using various supports; matching the patterns shown in~\rfig{fig-query-patterns}; and checking the existence of 14-cliques. We selected the patterns in \rfig{fig-query-patterns} to cover all the patterns used in \cite{fractal} and \cite{gminer}; note that patterns like triangles and empty squares are covered via applications like cliques and motifs.
Since G-Miner's pattern matching program is specific to labeled $p_2$ (in \rfig{fig-query-patterns}), we used labels on $p_2$ for all the systems to enable direct comparison. 
To match it on Orkut and Friendster graphs, which are unlabeled, we added synthetic labels (integers 1-6 as done in \cite{gminer}) with uniform probability.

\begin{table}
\small
\setlength{\tabcolsep}{0.8mm}
\renewcommand{\arraystretch}{.75}
  \begin{tabular}{c c | r | r r | r r}
  \multirow{2}{*}{App} & \multirow{2}{*}{$G$} & \multirow{2}{*}{\sysname{}}  & \multicolumn{2}{c|}{Arabesque} & \multicolumn{2}{c}{RStream} \\[0.02in]
   &  &   & ABQ-8 & ABQ-1 & RS-96 & RS-16 \\
    \midrule
  3-Motifs & MI  & \textbf{0.12}    & 158.05    & 39.05  & 51.83    & 252.74   \\
           & PA  & \textbf{3.10}    & 870.70    & 525.49 & 2685.45  & 2186.93  \\
           & OK  & \textbf{17.90}   & ---       &  ---   & /        & /        \\
           & FR  & \textbf{370.64}  & ---       &  ---   & /        & /        \\
    \midrule
  4-Motifs & MI  & \textbf{6.74}    & ---       &  ---   & /        & /        \\
           & PA  & \textbf{12.04}   & ---       &  ---   & /        & $\times$ \\
           & OK  & \textbf{6156.10} & ---       &  ---   & /        & /        \\
    \midrule
  2K-FSM   & MI  & \textbf{380.81}  & 3418.25   & 821.60 & $\times$ & ---      \\
  3K-FSM   & MI  & \textbf{279.74}  & 3520.82   & 784.27 & $\times$ & ---      \\
  4K-FSM   & MI  & \textbf{250.68}  & 3514.97   & 779.75 & $\times$ & ---      \\
    \midrule
  20K-FSM  & PA  & \textbf{859.41}  & ---             & ---    & 1757.69 & ---      \\
  21K-FSM  & PA  & \textbf{647.97}  & ---             & ---    & 1711.87 & ---      \\
  22K-FSM  & PA  & 507.56           & \textbf{342.63} & ---    & 1626.53 & ---      \\
  23K-FSM  & PA  & 402.57           & \textbf{299.12} & ---    & 1936.92 & ---      \\
    \midrule
  3-Cliques & MI & \textbf{0.05}    & 18.62     & 5.98   & 7.34     & 11.32    \\
            & PA & \textbf{0.59}    & 155.55    & 87.26  & 8.40     & 11.97    \\
            & OK & \textbf{13.75}   & ---       &  ---   & 986.20   & 1643.10  \\
            & FR & \textbf{296.99}  & ---       &  ---   & /        & /        \\
    \midrule
  4-Cliques & MI & \textbf{2.02}    & 1598.09   & 353.37 & 266.61   & ---      \\
            & PA & \textbf{0.90}    & 249.38    & 107.02 & 105.00   & 181.30   \\
            & OK & \textbf{281.47}  & ---       &  ---   & /        & /        \\
            & FR & \textbf{1337.77} & ---       &  ---   & /        & /        \\
    \midrule
  5-Cliques & MI & \textbf{89.60}   & $\times$  &  ---   & ---      & ---      \\
            & PA & \textbf{1.12}    & 352.64    & 122.09 & 145.00   & 237.90   \\
            & OK & \textbf{3182.56} & ---       &  ---   & /        & /        \\
            & FR & \textbf{4214.72} & ---       &  ---   & /        & /        \\
    \midrule
\end{tabular}
\caption{Execution times (in seconds) for \sysname{}, Arabesque~\cite{arabesque} and RStream~\cite{rstream}. \\ '$\times$' indicates the execution did not finish within 5 hours. \\ '---' indicates the system ran out of memory. \\ '/' indicates the system ran out of disk space.}
\label{bfs-comparison}
\end{table}

\begin{table}
\vspace{0.2in}
\renewcommand{\arraystretch}{.75}
\begin{tabular}{c c | r | r r}
  \multirow{2}{*}{App} & \multirow{2}{*}{$G$} & \multirow{2}{*}{\sysname{}}  & \multicolumn{2}{c}{Fractal}  \\[0.02in]
   & &  & FCL-8 & FCL-1 \\
    \midrule
  3-Motifs & MI    & \textbf{0.12}    & 22.13    &  17.11   \\
           & PA    & \textbf{3.10}    & 231.95   & 214.34   \\
           & OK    & \textbf{17.90}   & ---      &  ---     \\
           & FR    & \textbf{370.64}  & ---      &  ---     \\
    \midrule
  4-Motifs & MI    & \textbf{6.74}    & 78.66    & 420.67   \\
           & PA    & \textbf{12.04}   & 362.19   & 742.35   \\
           & OK    & \textbf{6156.10} & ---      &  ---     \\
    \midrule
  2K-FSM   & MI    & 380.81           & \textbf{154.47} & 675.98 \\
  3K-FSM   & MI    & 279.74           & \textbf{154.74} & 680.33 \\
  4K-FSM   & MI    & 250.68           & \textbf{144.34} & 663.26 \\
    \midrule
  20K-FSM  & PA    & \textbf{851.41}  & $\times$        & --- \\
  21K-FSM  & PA    & \textbf{647.97}  & $\times$        & --- \\
  22K-FSM  & PA    & \textbf{507.56}  & $\times$        & --- \\
  23K-FSM  & PA    & \textbf{402.57}  & 451.18          & --- \\
    \midrule
  3-Cliques & MI   & \textbf{0.05}    & 18.71    &  17.21   \\
            & PA   & \textbf{0.59}    & 232.60   & 216.76   \\
            & OK   & \textbf{13.75}   & ---      &  ---     \\
            & FR   & \textbf{296.99}  & ---      &  ---     \\
    \midrule
  4-Cliques & MI   & \textbf{2.02}     & 25.77   &  34.79  \\
            & PA   & \textbf{0.90}     & 237.64  & 224.50  \\
            & OK   & \textbf{281.47}   & ---     &  ---     \\
            & FR   & \textbf{1337.77}  & ---     &  ---     \\
    \midrule
  5-Cliques & MI   & \textbf{89.60}    & 181.30  &  904.65  \\
            & PA   & \textbf{1.12}     & 266.88  &  217.30  \\
            & OK   & \textbf{3182.56}  & ---     &  ---     \\
            & FR   & \textbf{4214.72}  & ---     &  ---     \\
    \midrule
    % swd
  Match $p_1$ & MI & \textbf{0.12}   & 24.76  & 36.02  \\
              & PA & \textbf{0.84}   & 235.72 & 189.03 \\
              & OK & \textbf{38.97}  & ---    & ---    \\
              & FR & \textbf{824.62} & ---    & ---    \\
    \midrule
    % triangle+2
  Match $p_2$ & MI & \textbf{0.03}   & 22.11  & 16.85  \\
              & PA & \textbf{1.07}   & 260.15 & 202.23 \\
              & OK & \textbf{474.09} & ---    & ---    \\
              & FR & \textbf{18.09}  & ---    & ---    \\
    \midrule
    % house
  Match $p_3$ & MI & \textbf{19.93}   & 181.76 & 1288.94 \\
              & PA & \textbf{13.41}   & 30.18  & 69.33   \\
              & OK & \textbf{13292.77}& ---    & ---     \\
    \midrule
    % house w diagonal
  Match $p_4$ & MI & \textbf{12.29}   & 120.99 & 789.81 \\
              & PA & \textbf{2.23}    & 25.58  & 21.63  \\
              & OK & \textbf{1569.73} & ---    & ---    \\
              & FR & \textbf{7057.40} & ---    & ---    \\
    \midrule
    % fullhouse
  Match $p_5$ & MI & \textbf{14.94}   & 56.51  & 345.35 \\
              & PA & \textbf{1.89}    & 25.30  & 17.39  \\
              & OK & \textbf{1381.03} & ---    & ---    \\
              & FR & \textbf{6726.51} & ---    & ---    \\
    \midrule
    % 6-vertex
  Match $p_6$ & MI & \textbf{65.26}   & $\times$ & $\times$ \\
              & PA & \textbf{27.94}   & 210.04   & 205.39   \\ \hline
\end{tabular}
\vspace{0.07in}
\caption{Execution times (in seconds) for \sysname{} and Fractal~\cite{fractal}. '---' indicates the system ran out of memory. \\ '$\times$' indicates the execution did not finish within 5 hours.}\label{fractal-comparison}
\end{table}

\subsection{Comparison with Breadth-First Enumeration}
\rtab{bfs-comparison} compares \sysname{}'s performance with Arabesque and RStream on motif-counting, clique-counting, and FSM (these systems do not support pattern matching). As we can see, \sysname{} outperforms the breadth-first systems by at least an order of magnitude on every application except FSM. RStream, despite being an out-of-core system, runs out of memory during FSM computations because of the massive amounts of aggregation information, as well as during 4- and 5-cliques on Mico where it could not handle the size of a single expansion step. 

It was interesting to observe that Arabesque performed better in single-node mode compared to 8-node configuration across all experiments except FSM on Patents, where it ran out of memory. This is because its breadth-first exploration generates large amounts of partial matches which must be synchronized across the entire cluster between supersteps, incurring high communication costs that impact its scalability. 

When support thresholds are high, Arabesque on 8 nodes computes FSM more quickly than \sysname{}. This is because its breadth-first strategy leverages high parallelism when there are few frequent patterns to explore and aggregate. 
However, this approach is sensitive to the support threshold, which stops Arabesque from scaling to lower threshold values where there are more frequent patterns. In these scenarios Arabesque simply fails due to the memory burden of maintaining the vast amount of intermediate matches and aggregation values. We suspect that even with more main memory per node, the intermediate computations (canonicality, isomorphism, etc.) for each individual match in Arabesque would significantly limit its performance. Since \sysname{} is pattern-aware, it only needs to maintain aggregation values for the patterns it is currently matching, allowing it to scale to inputs that yield many frequent patterns.

\newpage

\subsection{Comparison with Depth-First Enumeration}
\rtab{fractal-comparison} compares \sysname{}'s performance with Fractal on motif-counting, clique-counting, FSM, and pattern matching. As we can see, \sysname{} is faster than Fractal by at least an order of magnitude across most of the applications. For instance, 4-cliques on Patents finished in less than a second on \sysname{} whereas Fractal took over 200 seconds in both cluster and single-node configurations. 

Given equal resources (i.e., on a single node), FSM on Mico is up to 2.6$\times$ faster on \sysname{} compared to that on Fractal. Furthermore, \sysname{} scales to the larger dataset while Fractal does not. Even with 8 nodes, Fractal only outperforms \sysname{} on the small Mico graph, and cannot handle the Patents workload except for very high support thresholds, where there is less work to be done; there too, \sysname{} executes faster than Fractal.

Similar to Arabesque, Fractal's pattern-unawareness requires it to maintain global aggregation values throughout its computation. In FSM, the aggregation values consume $O(|V|)$ memory per vertex in each pattern in the worst case, and thus quickly become a scalability bottleneck. On the other hand, \sysname{} only needs to maintain aggregation values for the current patterns being matched, which allows it to achieve comparable performance and superior scalability while matching up to 15,817 patterns on Mico and 6,739 patterns on Patents.

\begin{table}
\begin{tabular}{c c | r | r r}
  \multirow{2}{*}{App} & \multirow{2}{*}{$G$} & \multirow{2}{*}{\sysname{}}  & \multicolumn{2}{c}{G-Miner}  \\[0.02in]
   &  &   & GM-8 & GM-1 \\
    \midrule
  3-Cliques & MI & \textbf{0.05}    &  3.79    & 3.86   \\
            & PA & \textbf{0.59}    &  7.91    & 8.93   \\
            & OK & \textbf{13.75}   &  44.26   & 62.65  \\
            & FR & \textbf{296.99}  &  /       &  /     \\
    \midrule
    % triangle+2
  Match $p_2$ & MI & \textbf{0.03}   & 3.67            & 3.95   \\
              & PA & \textbf{1.07}   & 6.84            & 9.80   \\
              & OK & 474.09          & \textbf{145.00} & 396.72 \\
              & FR & \textbf{18.09}  & /               & /      \\
    \midrule
\end{tabular}
\caption{Execution times (in seconds) for \sysname{} and G-Miner~\cite{gminer}. '/' indicates the system ran out of disk space.}\label{gminer-comparison}
\vspace{-0.3in}
\end{table}

\subsection{Comparison with Purpose-Built Algorithms}
\label{sec-comparison-purpose-built}
G-Miner is a general-purpose subgraph-centric system that targets expert users to implement the mining algorithms using a low-level subgraph data structure. Since expressing common mining algorithms requires domain expertise, we only evaluated the applications that are already implemented in G-Miner: 3-clique counting and pattern matching on $p_2$ (pattern matching for other patterns is not supported). This experiment serves to showcase how \sysname{} compares to custom algorithms for matching specific patterns.

\rtab{gminer-comparison} compares \sysname{}'s performance with G-Miner. As we can see, \sysname{} is 3$\times$ to 77$\times$ faster than G-Miner when counting 3-cliques even though G-Miner implements an algorithm designed specifically to count 3-cliques. When matching $p_2$, \sysname{} is 6$\times$ to 131$\times$ faster on Mico and Patents. On Orkut, however, G-Miner performs better on finding $p_2$; this is because G-Miner indexes vertices by labels when preprocessing the data graph, whereas \sysname{} discovers labels dynamically. Due to these indexes, G-Miner could not handle Friendster even with 240GB disk space on the cluster.

\begin{table}[t]
  \begin{tabular}{c | r | r | r}
    \multirow{2}{*}{$G$} & Existence & Anti-Vertex & Anti-Edge \\
     & 14-Clique & \multicolumn{1}{c |}{$p_7$} & \multicolumn{1}{c}{$p_8$} \\
    \midrule
     MI & 0.07  & 0.65   & 6.92    \\
     PA & 3.95  & 0.67   & 1.69    \\
     OK & 4.08  & 56.06  & 879.01  \\
     FR & 50.39 & 470.21 & 4017.15 \\
    \midrule
  \end{tabular}
  \caption{\sysname{} execution times (in seconds) for matching with an anti-vertex ($p_7$), matching with an anti-edge ($p_8$), and 14-clique existence query.}
  \label{usonly}
  \vspace{-0.25in}
\end{table}

\subsection{Mining with Constraints in \sysname{}}
We evaluate \sysname{} on mining tasks with structural constraints. We match a pattern containing an anti-vertex ($p_7$), one containing an anti-edge ($p_8$), and perform an existence query of a 14-clique. The results are show in~\rtab{usonly}.

\paragraph{Mining with Anti-Vertices.}
Pattern $p_7$ expresses a maximal clique of size 3 (triangle) using a fully-connected anti-vertex, i.e., it matches all triangles that are not contained in a 4-clique. While satisfying the anti-vertex constraint requires computing set-intersections across all vertices of the triangle, \sysname{} takes less than a minute on Orkut, and under eight minutes on the billion scale Friendster graph.

\paragraph{Mining with Anti-Edges.}
Pattern $p_8$ represents a vertex-induced chordal square using an anti-edge constraint. Satisfying the anti-edge constraint is computationally demanding, since it requires computing set differences of adjacency lists, which is twice as many operations as the sum of the adjacency list sizes. Nevertheless, \sysname{} still easily completes it on all the datasets.

\paragraph{Existence Query.}
The goal of this query is to determine whether a 14-clique exists in the data graph. \sysname{} stops exploration immediately after finding an instance of 14-clique. We observe that Patents and Orkut performed similarly; this is because the rarer the target pattern for an existence query, the longer it takes to find it. Patents does not contain a 14-clique, so the entire graph was searched, but in the much larger and denser Orkut graph, a 14-clique gets found quickly during exploration. Friendster is both large and sparse, and hence, 14-cliques are rare. Furthermore, since 14-clique is a large pattern, several partial explorations do not lead to a complete 14-clique.

\subsection{\sysname{}'s Pattern-Aware Runtime}
\paragraph{Benefits of Symmetry Breaking}
Symmetry breaking is a well-studied technique for subgraph matching that \sysname{} uses to guide its graph exploration. However, recent systems like Fractal~\cite{fractal} and AutoMine~\cite{automine} are not fully pattern-aware and do not leverage symmetry breaking for common graph mining use cases. We showcase the importance of symmetry breaking in \sysname{} by disabling it and running 4-motifs and FSM with low support thresholds. These are expensive subgraph matching workloads: 4-motifs contains complex patterns with many matches and FSM involves a large number of patterns to match. \rfig{nosymm-comparison} summarizes the results.

We observe that symmetry breaking improves performance by an order of magnitude for 4-motifs on Mico and Patents. Orkut 4-motifs without symmetry breaking did not even finish matching even a single size 4 pattern within 5 hours. This shows the importance of symmetry breaking when scaling to large patterns and large datasets. For instance, Orkut contains over 22 trillion \emph{unique} vertex-induced 4-stars, and so without symmetry breaking, the system must process six times that many matches (a 4-star's automorphisms are the permutations of its 3 endpoints: resulting in $3! = 6$ automorphic subgraphs).

FSM achieves ~3$\times$ performance improvement through symmetry breaking. This is because with symmetry breaking, FSM's expensive aggregation values are only written to once per unique match in the data graph, whereas the naive approach without symmetry breaking would incur dozens of redundant write (and read) accesses per unique match.

\begin{figure}
  \includegraphics[width=\textwidth/2]{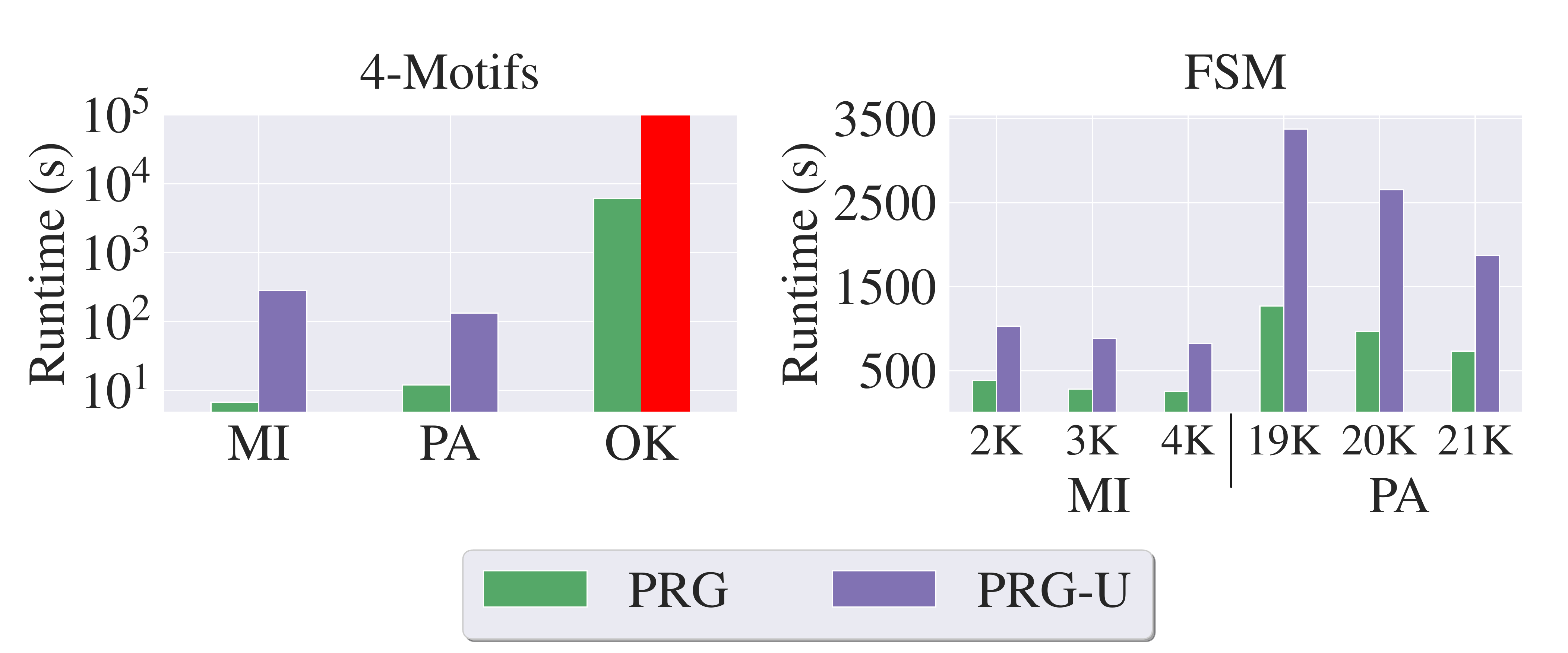}
  \caption{Execution times (in seconds) for \sysname{} with (PRG) and without (PRG-U) symmetry breaking. PRG-U could not finish matching any of the 4-motif patterns on Orkut within 5 hours.}
  \label{nosymm-comparison}
\end{figure}

\begin{figure}[t]
  \includegraphics[width=\textwidth/2]{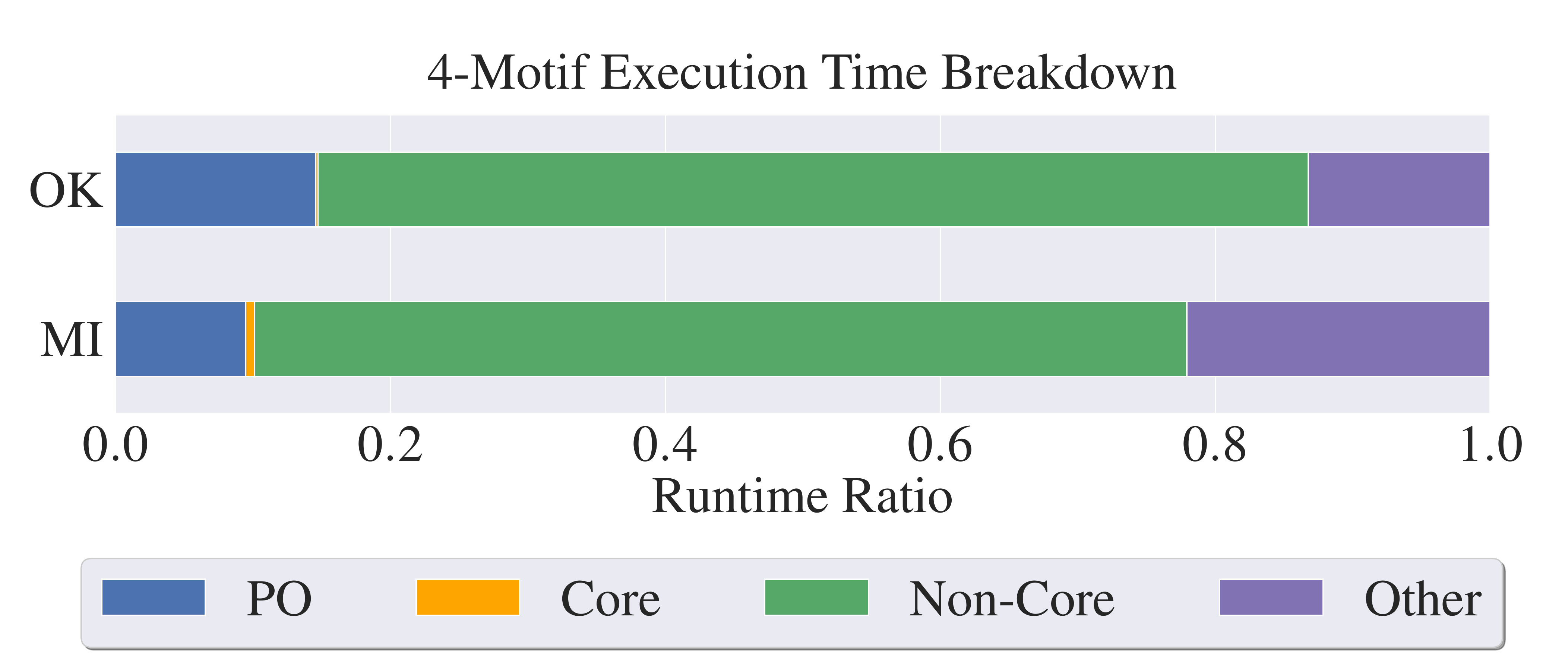}
  \caption{\sysname{} 4-motif execution time breakdown.}
  \label{plot-ratios}
\end{figure}

\paragraph{Breakdown on Mining Time.} 
\rfig{plot-ratios} shows the ratio of time spent in each stage of matching during 4-motif execution: finding the range of sorted candidate sets that meet the pattern's partial order (PO), performing adjacency list intersections and differences to match the pattern core (Core) and finally, intersecting the adjacency lists of the pattern core to complete the match (Non-Core). Some time is also spent on the other requirements of matching, for example, fetching adjacency lists and mapping vertices (Other).

We observe that the majority of execution is spent intersecting adjacency lists of candidate vertices to complete matches. In comparison to the overall execution time, matching the core pattern is insignificant. This is because the core pattern is matched according to all valid total orderings of its vertices, and hence, the traversal is fully guided. In contrast, the non-core vertices may or may not be ordered with respect to each other, and with respect to the core vertices; so the runtime usually has less guidance when exploring the graph. Furthermore, in most patterns the core is small and involves fewer intersections than the non-core component.

\subsection{System Characteristics}
\paragraph{Scalability.} We study \sysname{}'s scalability by matching pattern $p_1$ on Orkut using \texttt{c5.metal} instance. Note that we do not perform a COST analysis~\cite{costperformance} with this experiment since we already compared \sysname{} with optimized algorithms in \rsec{sec-comparison-purpose-built}, and state-of-the-art serial pattern matching solutions like~\cite{vf2, turboiso} performed much slower than our single threaded execution. 

\rfig{scalability-plots} shows how \sysname{} scales as number of threads increase from 1 to 96. As we can see, \sysname{} scales linearly until 48 threads, after which speedups increase gradually. This is mainly because \texttt{c5.metal} has 48 physical cores, and scheduling beyond 48 threads happens with hyper-threading. We verified this effect by alternating how threads get scheduled across different cores; the dashed lines in \rfig{scalability-plots} show speedups when every pair of \sysname{} threads is pinned to two logical CPUs on one physical CPU. As we can see, with 48 threads but only 24 physical cores, \sysname{} only achieves a 30$\times$ speedup, whereas with 48 physical cores it achieves a 41$\times$ speedup. Since pattern exploration involves continuous random memory accesses throughout execution, hyper-threading helps in hiding memory latencies only up to an extent. \rfig{plot-utilization} verifies this, as memory bandwidth grows considerably higher when using more cores, though CPU utilization remains similar.
 
We observe that speedups also decline slightly between 24 cores and 48 cores. This is because \texttt{c5.metal} has two NUMA nodes, each allocated to 24 physical cores. We measured remote memory accesses to observe the NUMA effects: when running on 48 cores, cross-numa memory traffic was 86GB as opposed to only 4.9MB when running on 24 cores. 

\begin{figure}[t]
  \subfloat{%
    \includegraphics[width=\textwidth/4-9mm]{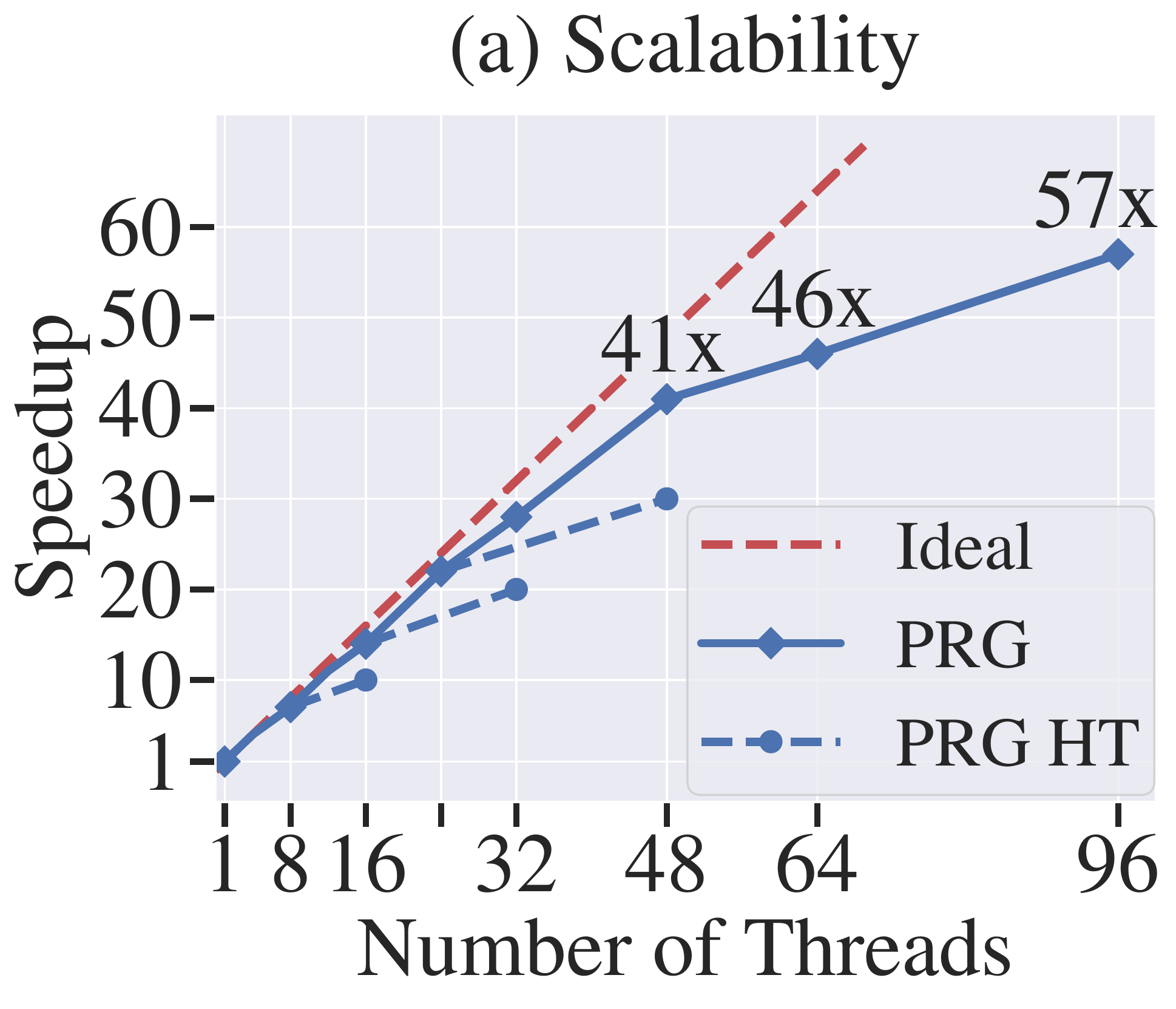}
    \label{scalability-plots}
  }
  \hfill
  \subfloat{%
    \includegraphics[width=\textwidth/4+1mm]{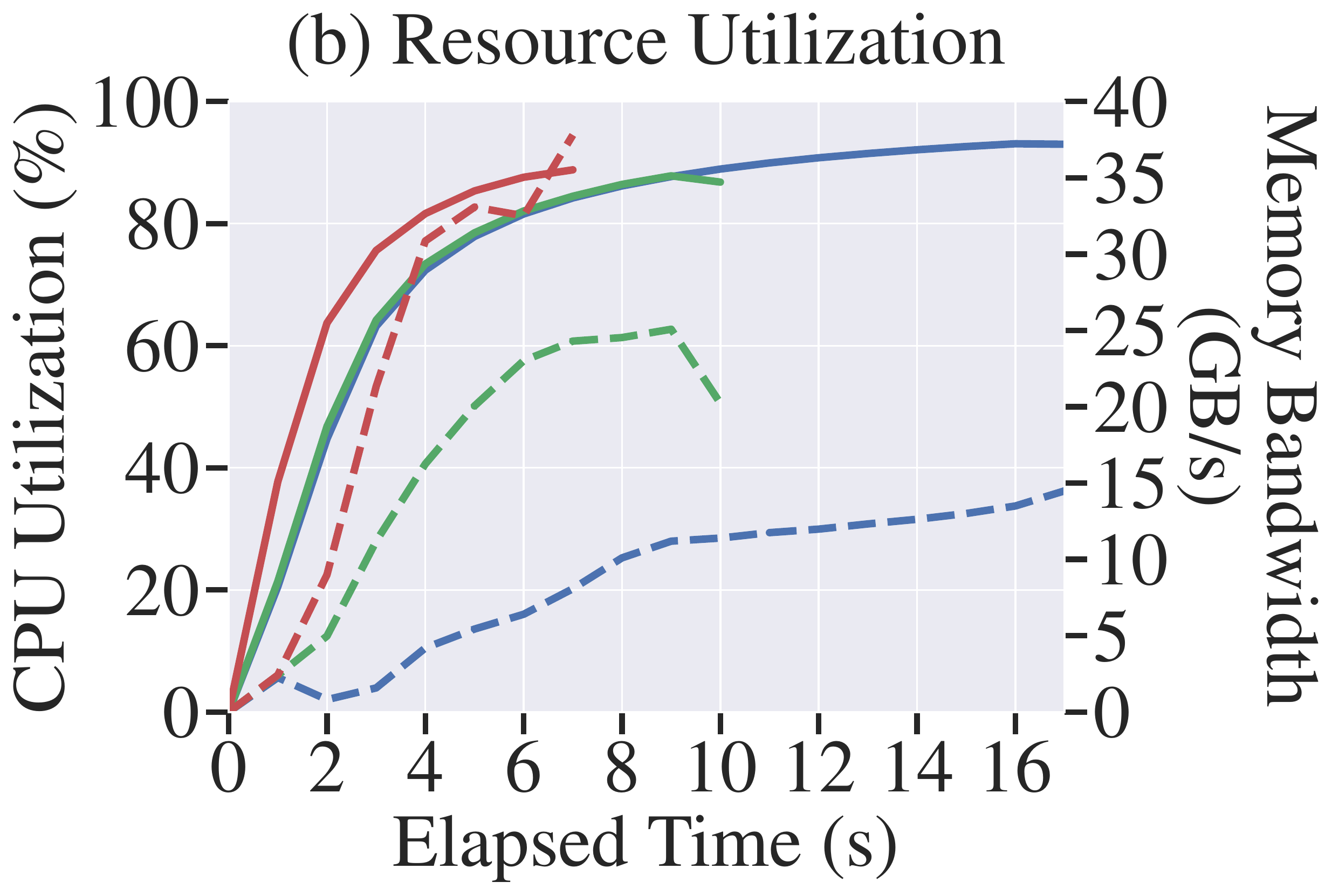}
    \label{plot-utilization}
  }
  \caption{(a) Scalability (PRG HT = hyper-threaded). \\ (b) CPU utilization (solid) and memory bandwidth (dashed) for 24 cores (blue), 47 cores (green) and 94 cores (red).}
  \label{scalability-utilization-plots}
\end{figure}

\paragraph{Resource Utilization.} \rfig{plot-utilization} shows CPU utilization and memory bandwidth consumed by \sysname{} while matching $p_1$ on Orkut on \texttt{c5.metal} with 24, 47, and 94 threads. We reserve a core for profiling to avoid its overhead. We observe that \sysname{} maintains high CPU utilization throughout its execution. The memory bandwidth curve increases over time; as high degree vertices finish processing, low degree vertices do less computation and incur more memory accesses as they get processed. 

\rfig{fig-memory} compares the peak memory usage for \sysname{} and other systems. For distributed systems we report the sum of all nodes' peak memory. \sysname{} consistently uses less memory than all the systems, mainly because of its direct pattern-aware exploration strategy. It is interesting to note that changing the pattern size in cliques and motifs does not impact \sysname{}'s memory usage. The usage is high for FSM compared to other applications due to large domain maps for support calculation. 

\paragraph{Load Balancing.}
Since \sysname{} threads dynamically pick up tasks as they become free, we observe near-zero load imbalance while matching $p_1$ across all our datasets. The difference between times taken by threads to finish all of their work was only up to 71 ms. 

\section{Related Work}
There has been a variety of research to develop efficient graph mining solutions. 
To the best of our knowledge, \sysname{} is the first general-purpose graph mining system to leverage pattern-awareness in its programming and processing models. 

\begin{figure}[t]
\centering
  \includegraphics[width=3.3in]{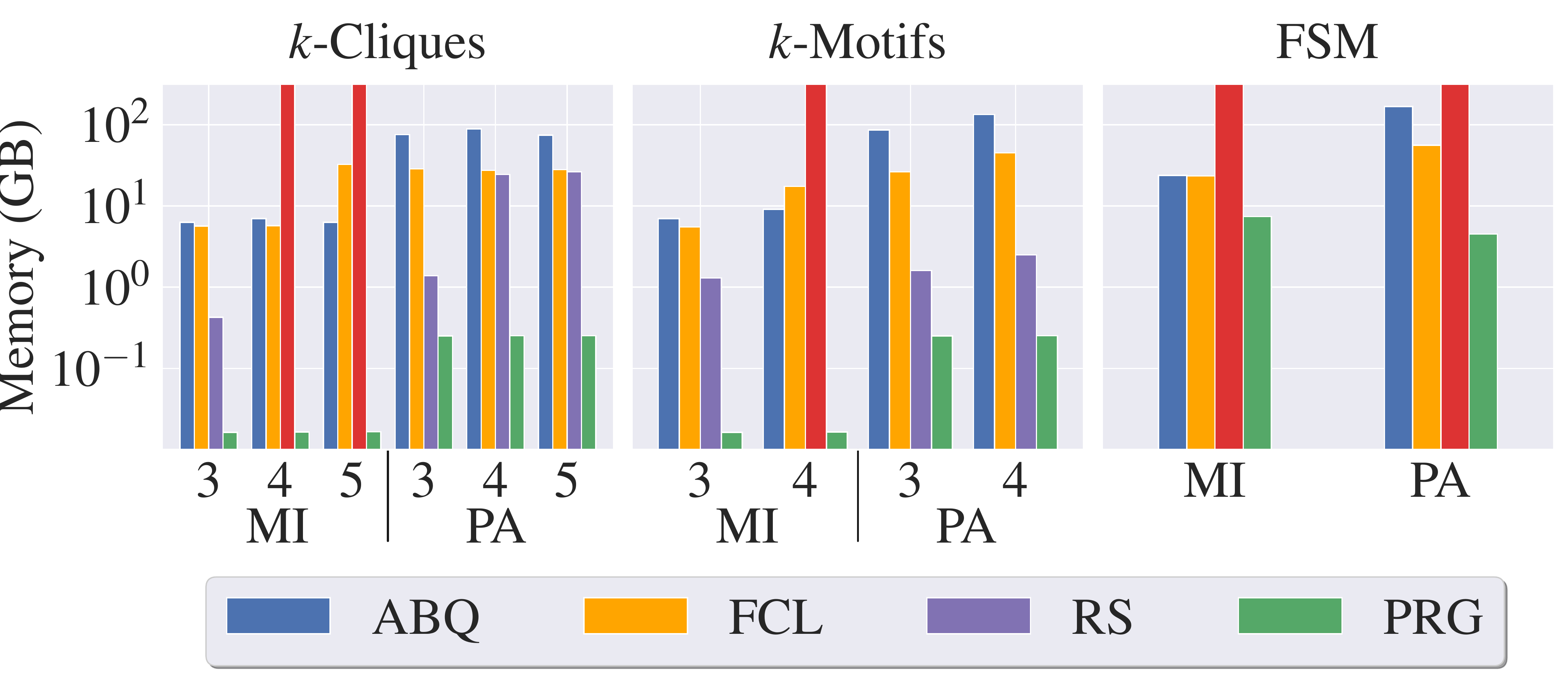}
  \caption{Peak memory usage of different systems across various applications. Tall red bars represent RStream out of memory errors.}
  \label{fig-memory}
\end{figure}

\paragraph{General-Purpose Graph Mining Systems.}
Several general-purpose graph mining systems have been developed~\cite{arabesque,rstream,asap,fractal,automine,gminer}. 
Arabesque~\cite{arabesque} is a distributed graph mining system that follows a filter-process model developed on top of map-reduce. It proposed the ``Think Like an Embedding'' processing model. 
Fractal~\cite{fractal} extends this to the concept of \emph{fractoids}, which expose parts of the user program to the system; in conjunction with depth-first exploration, fractoids allow the system to more intelligently plan its execution. 
G-Miner~\cite{gminer} is a task-oriented distributed graph mining system that enables building custom graph mining use cases using a distributed task queue.
RStream~\cite{rstream} is a single machine out-of-core graph mining system that leverages SSDs to store intermediate solutions. It uses relational algebra to express mining tasks as table joins. 
AutoMine~\cite{automine} is a recent single-machine system that generates efficient code to match patterns for common graph mining tasks. 
As discussed in \rsec{sec-issueswithgraphminingsystems}, none of these systems are fully pattern-aware, the way \sysname{} is: these systems perform unnecessary explorations and computations, require large memory (or storage) capacity, and lack the ability to easily express mining tasks at a high level. While Fractal uses symmetry breaking for pattern matching use case, other applications like FSM and motif counting are not guided by symmetry breaking, and hence they end up performing unnecessary explorations. Similarly, AutoMine also does not employ symmetry breaking for any of the use cases, requiring users to filter duplicate matches by individually examining every single match when enumerating patterns. Lack of full pattern-awareness not only makes these systems slower, but also limits their applicability to more complex mining use cases. 
Finally, ASAP~\cite{asap} is a programmable distributed system for approximate graph mining where users write programs based on sampling edges and vertices to reason about the probabilistic counts of patterns.

\paragraph{Purpose-Built Graph Mining Solutions.}
These works efficiently perform specific graph mining tasks.
ApproxG~\cite{approxg} is an efficient system for computing approximate graphlet (motif) counts with accuracy guarantees. \cite{efficient-graphlets} uses combinatorial arguments to obtain counts for size 3 and 4 motifs after counting smaller motifs. \cite{kclist} efficiently lists $k$-cliques in sparse graphs and \cite{maxkplex} is aimed at $k$-plexes which are clique-like structures. GraMi~\cite{mico} leverages anti-monotonicity for FSM on a single machine while ScaleMine~\cite{scalemine} is a distributed system for FSM that uses efficiently computable approximate stats to inform its graph exploration. \cite{distgraph} is also a distributed system focusing on FSM. \cite{maverick, exemplars} are recent works aimed at analyzing small graphs whose edges have large attribute sets.

Several systems aim to perform efficient pattern matching. 
OPT~\cite{opt} is a fast single-machine out-of-core triangle-counting system whose techniques are generalized by DualSim~\cite{dualsim} to match arbitrary patterns. 
\cite{tcshun} proposes several provably cache-friendly parallel triangle-counting algorithms which provide order-of-magnitude speedups over previous algorithms. DistTC~\cite{disttc} presents a distributed triangle-counting technique that leverages a novel graph partitioning strategy to count triangles with minimal communication overhead.

\cite{seed} is a distributed map-reduce based pattern matching system that first finds small patterns and joins them into large ones. QFrag~\cite{qfrag} is another map-reduce based distributed pattern matching system that focuses on searching graphs for large patterns using the TurboISO~\cite{turboiso} algorithm. PruneJuice~\cite{prunejuice} is a distributed pattern matching system that focuses on pruning data graph vertices that cannot contribute to a match. 
\cite{daf} is a scalable subgraph isomorphism algorithm while TurboFlux~\cite{turboflux} performs pattern matching on dynamically changing data graphs. \cite{hybridoptimizer} presents a pattern matching plan optimizer incorporated in Graphflow~\cite{graphflow} that uses both binary and multi-way joins. \cite{pgxdasync} is a resource-aware distributed graph querying system for property graphs. 

\paragraph{Graph Processing Systems.}
Several works enable processing static and dynamic graphs~\cite{pregel,powergraph,galois,ligra,graphbolt,chaos,graphx,gemini,kickstarter,aspen,graphone,aspire,coral,gps,pgxd}. These systems typically compute values on vertices and edges rather than analyzing substructures in graphs. They decompose computation at vertex and edge level, which is not suitable for graph mining use cases.

\section{Conclusion}
We presented \sysname{}, a pattern-aware graph mining system that efficiently explores subgraph structures of interest, and scales to complex graph mining tasks on large graphs. \sysname{} uses `pattern-based programming' that treats \mbox{\emph{patterns}} as first class constructs. We further introduced two novel abstractions: \textsc{Anti-Edge} and \textsc{Anti-Vertex}, that express advanced structural constraints on patterns to be matched. This allows users to directly operate on patterns and easily express complex mining use cases as `pattern programs' on \sysname{}. 

Our extensive evaluation showed that \sysname{} outperforms the existing state-of-the-art by several orders of magnitude, even when it has access to up to 8$\times$ fewer CPU cores. Furthermore, \sysname{} successfully handles resource-intensive graph mining tasks on billion-scale graphs on a single machine, while the state-of-the-art fails even with a cluster of 8 such machines or access to large SSDs. 

\section*{Acknowledgments}
We would like to thank our shepherd Aleksandar Prokopec and the anonymous reviewers for their valuable and thorough feedback. This work is supported by the Natural Sciences and Engineering Research Council of Canada.

\balance
\bibliographystyle{plain}
\setlength{\bibsep}{1pt plus 0.3ex}
\bibliography{paper}

\begin{thebibliography}{10}

\bibitem{scalemine}
Ehab Abdelhamid, Ibrahim Abdelaziz, Panos Kalnis, Zuhair Khayyat, and Fuad
  Jamour.
\newblock ScaleMine: Scalable Parallel Frequent Subgraph Mining in a Single
  Large Graph.
\newblock In {\em Proceedings of the International Conference for High
  Performance Computing, Networking, Storage and Analysis (SC '16)}, pages
  61:1--61:12, 2016.

\bibitem{efficient-graphlets}
Nesreen~K. Ahmed, Jennifer Neville, Ryan~A. Rossi, and Nick Duffield.
\newblock Efficient Graphlet Counting for Large Networks.
\newblock In {\em IEEE International Conference on Data Mining (ICDM '15)},
  pages 1--10, 2015.

\bibitem{bearman-chains}
Peter~S. Bearman, James Moody, and Katherine Stovel.
\newblock Chains of Affection: The Structure of Adolescent Romantic and Sexual
  Networks.
\newblock {\em American Journal of Sociology}, 110(1):44--91, 2004.

\bibitem{maxkplex}
Devora Berlowitz, Sara Cohen, and Benny Kimelfeld.
\newblock Efficient Enumeration of Maximal k-Plexes.
\newblock In {\em Proceedings of the ACM International Conference on Management
  of Data (SIGMOD '15)}, pages 431--444, 2015.

\bibitem{postponecart}
Fei Bi, Lijun Chang, Xuemin Lin, Lu~Qin, and Wenjie Zhang.
\newblock Efficient Subgraph Matching by Postponing Cartesian Products.
\newblock In {\em Proceedings of the ACM International Conference on Management
  of Data (SIGMOD '16)}, pages 1199--1214, 2016.

\bibitem{fsm-mni}
Bjorn Bringmann and Siegfried Nijssen.
\newblock What Is Frequent in a Single Graph?
\newblock In {\em Advances in Knowledge Discovery and Data Mining: 12th
  Pacific-Asia Conference}, volume 5012, pages 858--863, 2008.

\bibitem{roaring}
Samy Chambi, Daniel Lemire, Owen Kaser, and Robert Godin.
\newblock Better Bitmap Performance with Roaring Bitmaps.
\newblock {\em Software: Practice and Experience}, 46(5):709--719, 2016.

\bibitem{gminer}
Hongzhi Chen, Miao Liu, Yunjian Zhao, Xiao Yan, Da~Yan, and James Cheng.
\newblock G-Miner: An Efficient Task-oriented Graph Mining System.
\newblock In {\em Proceedings of the European Conference on Computer Systems
  (EuroSys '18)}, pages 32:1--32:12, 2018.

\bibitem{fsmcompvis}
Wei-Ta Chu and Ming-Hung Tsai.
\newblock Visual pattern discovery for architecture image classification and
  product image search.
\newblock In {\em Proceedings of the ACM International Conference on Multimedia
  Retrieval (ICMR '12)}, pages 1--8, 2012.

\bibitem{kclist}
Maximilien Danisch, Oana Balalau, and Mauro Sozio.
\newblock Listing K-cliques in Sparse Real-World Graphs*.
\newblock In {\em Proceedings of the World Wide Web Conference (WWW '18)},
  pages 589--598, 2018.

\bibitem{aspen}
Laxman Dhulipala, Guy~E Blelloch, and Julian Shun.
\newblock Low-Latency Graph Streaming Using Compressed Purely-Functional Trees.
\newblock In {\em Proceedings of the ACM SIGPLAN Conference on Programming
  Language Design and Implementation (PLDI '19)}, pages 918--934, 2019.

\bibitem{fractal}
Vinicius Dias, Carlos H.~C. Teixeira, Dorgival Guedes, Wagner Meira, and
  Srinivasan Parthasarathy.
\newblock Fractal: A General-Purpose Graph Pattern Mining System.
\newblock In {\em Proceedings of the ACM International Conference on Management
  of Data (SIGMOD '19)}, pages 1357--1374, 2019.

\bibitem{mico}
Mohammed Elseidy, Ehab Abdelhamid, Spiros Skiadopoulos, and Panos Kalnis.
\newblock GraMi: Frequent Subgraph and Pattern Mining in a Single Large Graph.
\newblock In {\em Proceedings of the VLDB Endowment (PVLDB '14)}, pages
  517--528, 2014.

\bibitem{powergraph}
Joseph~E. Gonzalez, Yucheng Low, Haijie Gu, Danny Bickson, and Carlos Guestrin.
\newblock PowerGraph: Distributed Graph-parallel Computation on Natural Graphs.
\newblock In {\em Proceedings of the USENIX Conference on Operating Systems
  Design and Implementation (OSDI '12)}, pages 17--30, 2012.

\bibitem{graphx}
Joseph~E. Gonzalez, Reynold~S. Xin, Ankur Dave, Daniel Crankshaw, Michael~J.
  Franklin, and Ion Stoica.
\newblock {GraphX}: Graph Processing in a Distributed Dataflow Framework.
\newblock In {\em Proceedings of the USENIX Conference on Operating Systems
  Design and Implementation (OSDI '14)}, pages 599--613, 2014.

\bibitem{po}
Joshua~A. Grochow and Manolis Kellis.
\newblock Network Motif Discovery Using Subgraph Enumeration and
  Symmetry-Breaking.
\newblock In {\em Research in Computational Molecular Biology}, pages 92--106,
  2007.

\bibitem{patents}
Bronwyn Hall, Adam Jaffe, and Manuel Trajtenberg.
\newblock The NBER Patent Citation Data File: Lessons, Insights and
  Methodological Tools.
\newblock {\em NBER Working Paper 8498}, 2001.

\bibitem{daf}
Myoungji Han, Hyunjoon Kim, Geonmo Gu, Kunsoo Park, and Wook-Shin Han.
\newblock Efficient Subgraph Matching: Harmonizing Dynamic Programming,
  Adaptive Matching Order, and Failing Set Together.
\newblock In {\em Proceedings of the ACM International Conference on Management
  of Data (SIGMOD '19)}, pages 1429--1446, 2019.

\bibitem{turboiso}
Wook-Shin Han, Jinsoo Lee, and Jeong-Hoon Lee.
\newblock TurboISO: Towards Ultrafast and Robust Subgraph Isomorphism Search in
  Large Graph Databases.
\newblock In {\em Proceedings of the ACM International Conference on Management
  of Data (SIGMOD '13)}, pages 337--348, 2013.

\bibitem{disttc}
Loc Hoang, Vishwesh Jatala, Xuhao Chen, Udit Agarwal, Roshan Dathathri,
  Gurbinder Gill, and Keshav Pingali.
\newblock DistTC: High Performance Distributed Triangle Counting.
\newblock In {\em IEEE High Performance Extreme Computing Conference (HPEC
  '19)}, pages 1--7, 2019.

\bibitem{pgxd}
Sungpack Hong, Siegfried Depner, Thomas Manhardt, Jan Van Der~Lugt, Merijn
  Verstraaten, and Hassan Chafi.
\newblock PGX.D: a fast distributed graph processing engine.
\newblock In {\em Proceedings of the International Conference for High
  Performance Computing, Networking, Storage and Analysis (SC '15)}, pages
  1--12, 2015.

\bibitem{asap}
Anand~Padmanabha Iyer, Zaoxing Liu, Xin Jin, Shivaram Venkataraman, Vladimir
  Braverman, and Ion Stoica.
\newblock {ASAP}: Fast, Approximate Graph Pattern Mining at Scale.
\newblock In {\em Proceedings of the USENIX Symposium on Operating Systems
  Design and Implementation (OSDI '18)}, pages 745--761, Carlsbad, CA, 2018.

\bibitem{graphflow}
Chathura Kankanamge, Siddhartha Sahu, Amine Mhedbhi, Jeremy Chen, and Semih
  Salihoglu.
\newblock Graphflow: An Active Graph Database.
\newblock In {\em Proceedings of the ACM International Conference on Management
  of Data (SIGMOD '17)}, pages 1695--1698, 2017.

\bibitem{dualsim}
Hyeonji Kim, Juneyoung Lee, Sourav~S. Bhowmick, Wook-Shin Han, JeongHoon Lee,
  Seongyun Ko, and Moath~H.A. Jarrah.
\newblock DUALSIM: Parallel Subgraph Enumeration in a Massive Graph on a Single
  Machine.
\newblock In {\em Proceedings of the ACM International Conference on Management
  of Data (SIGMOD '16)}, pages 1231--1245, 2016.

\bibitem{opt}
Jinha Kim, Wook-Shin Han, Sangyeon Lee, Kyungyeol Park, and Hwanjo Yu.
\newblock OPT: A New Framework for Overlapped and Parallel Triangulation in
  Large-scale Graphs.
\newblock In {\em Proceedings of the ACM International Conference on Management
  of Data (SIGMOD '14)}, pages 637--648, 2014.

\bibitem{turboflux}
Kyoungmin Kim, In~Seo, Wook-Shin Han, Jeong-Hoon Lee, Sungpack Hong, Hassan
  Chafi, Hyungyu Shin, and Geonhwa Jeong.
\newblock TurboFlux: A Fast Continuous Subgraph Matching System for Streaming
  Graph Data.
\newblock In {\em Proceedings of the ACM International Conference on Management
  of Data (SIGMOD '18)}, pages 411--426, 2018.

\bibitem{globalclustering}
Anton Korshunov, Ivan Beloborodov, Nazar Buzun, Valeriy Avanesov, Roman
  Pastukhov, Kyrylo Chykhradze, Ilya Kozlov, Andrey Gomzin, Ivan Andrianov,
  Andrey Sysoev, Stepan Ipatov, Ilya Filonenko, Christina Chuprina, Denis
  Turdakov, and Sergey Kuznetsov.
\newblock Social Network Analysis: Methods and Applications.
\newblock In {\em Proceedings of the Institute for System Programming of RAS},
  pages 439--456, 2014.

\bibitem{cliquescompchem}
Frederick~S. Kuhl, Gordon~M. Crippen, and Donald~K. Friesen.
\newblock A combinatorial algorithm for calculating ligand binding.
\newblock {\em Journal of Computational Chemistry}, 5(1):24--34, 1984.

\bibitem{graphone}
Pradeep Kumar and H~Howie Huang.
\newblock GraphOne: A Data Store for Real-Time Analytics on Evolving Graphs.
\newblock In {\em Proceedings of the USENIX Conference on File and Storage
  Technologies (FAST '19)}, pages 249--263, 2019.

\bibitem{mis-support}
Michihiro Kuramochi and George Karypis.
\newblock Finding Frequent Patterns in a Large Sparse Graph*.
\newblock {\em Data Mining and Knowledge Discovery}, 11(3):243--271, 2005.

\bibitem{seed}
Longbin Lai, Lu~Qin, Xuemin Lin, Ying Zhang, Lijun Chang, and Shiyu Yang.
\newblock Scalable Distributed Subgraph Enumeration.
\newblock In {\em Proceedings of the VLDB Endowment (PVLDB '16)}, pages
  217--228, 2016.

\bibitem{pregel}
Grzegorz Malewicz, Matthew~H. Austern, Aart J.~C. Bik, James~C. Dehnert, Ilan
  Horn, Naty Leiser, Grzegorz Czajkowski, and Google Inc.
\newblock Pregel: A System for Large-Scale Graph Processing.
\newblock In {\em Proceedings of the ACM International Conference on Management
  of Data (SIGMOD '10)}, pages 135--146, 2010.

\bibitem{graphbolt}
Mugilan Mariappan and Keval Vora.
\newblock GraphBolt: Dependency-Driven Synchronous Processing of Streaming
  Graphs.
\newblock In {\em Proceedings of the European Conference on Computer Systems
  (EuroSys '19)}, pages 25:1--25:16, 2019.

\bibitem{automine}
Daniel Mawhirter and Bo~Wu.
\newblock AutoMine: Harmonizing High-level Abstraction and High Performance for
  Graph Mining.
\newblock In {\em Proceedings of the ACM Symposium on Operating Systems
  Principles (SOSP '19)}, pages 509--523, 2019.

\bibitem{approxg}
Daniel Mawhirter, Bo~Wu, Dinesh Mehta, and Chao Ai.
\newblock ApproxG: Fast Approximate Parallel Graphlet Counting Through Accuracy
  Control.
\newblock In {\em IEEE/ACM International Symposium on Cluster, Cloud and Grid
  Computing (CCGRID '18)}, pages 533--542, 2018.

\bibitem{sna-police}
Jean McGloin and David Kirk.
\newblock An Overview of Social Network Analysis.
\newblock {\em Journal of Criminal Justice Education}, 21:169--181, 2010.

\bibitem{costperformance}
Frank McSherry, Michael Isard, and Derek~G. Murray.
\newblock Scalability! But at what {COST}?
\newblock In {\em Workshop on Hot Topics in Operating Systems (HotOS {XV})},
  2015.

\bibitem{merged-support}
Jinghan Meng and Yi-cheng Tu.
\newblock Flexible and Feasible Support Measures for Mining Frequent Patterns
  in Large Labeled Graphs.
\newblock In {\em Proceedings of the ACM International Conference on Management
  of Data (SIGMOD '17)}, page 391–402, 2017.

\bibitem{app-specific-support}
Pieter Meysman, Yvan Saeys, Ehsan Sabaghian, Wout Bittremieux, Yves Van~de
  Peer, Bart Goethals, and Kris Laukens.
\newblock Discovery of Significantly Enriched Subgraphs Associated with
  Selected Vertices in a Single Graph.
\newblock In {\em Proceedings of the International Workshop on Data Mining in
  Bioinformatics (BIOKDD '15)}, volume~15, pages 1--8, 2015.

\bibitem{hybridoptimizer}
Amine Mhedhbi and Semih Salihoglu.
\newblock Optimizing Subgraph Queries by Combining Binary and Worst-Case
  Optimal Joins.
\newblock In {\em Proceedings of the VLDB Endowment (PVLDB '19)}, page
  1692–1704, 2019.

\bibitem{motifsbiology}
Ron Milo, Shai Shen-Orr, Shalev Itzkovitz, N~Kashtan, Dmitri Chklovskii, and
  Uri Alon.
\newblock Network Motifs: Simple Building Blocks of Complex Networks.
\newblock {\em Science}, 298:824--7, 2002.

\bibitem{galois}
Donald Nguyen, Andrew Lenharth, and Keshav Pingali.
\newblock A Lightweight Infrastructure for Graph Analytics.
\newblock In {\em Proceedings of the ACM Symposium on Operating Systems
  Principles (SOSP '13)}, pages 456--471, 2013.

\bibitem{prunejuice}
Tahsin Reza, Matei Ripeanu, Nicolas Tripoul, Geoffrey Sanders, and Roger
  Pearce.
\newblock PruneJuice: Pruning Trillion-edge Graphs to a Precise
  Pattern-matching Solution.
\newblock In {\em Proceedings of the International Conference for High
  Performance Computing, Networking, Storage, and Analysis (SC '18)}, pages
  21:1--21:17, 2018.

\bibitem{pgxdasync}
Nicholas~P. Roth, Vasileios Trigonakis, Sungpack Hong, Hassan Chafi, Anthony
  Potter, Boris Motik, and Ian Horrocks.
\newblock PGX.D/Async: A Scalable Distributed Graph Pattern Matching Engine.
\newblock In {\em Proceedings of the International Workshop on Graph
  Data-Management Experiences \& Systems (GRADES '17)}, 2017.

\bibitem{chaos}
Amitabha Roy, Laurent Bindschaedler, Jasmina Malicevic, and Willy Zwaenepoel.
\newblock Chaos: Scale-out Graph Processing from Secondary Storage.
\newblock In {\em Proceedings of the ACM Symposium on Operating Systems
  Principles (SOSP '15)}, pages 410--424, 2015.

\bibitem{gps}
Semih Salihoglu and Jennifer Widom.
\newblock GPS: A Graph Processing System.
\newblock In {\em Proceedings of the International Conference on Scientific and
  Statistical Database Management (SSDBM '13)}, 2013.

\bibitem{qfrag}
Marco Serafini, Gianmarco De~Francisci~Morales, and Georgos Siganos.
\newblock QFrag: Distributed Graph Search via Subgraph Isomorphism.
\newblock In {\em Proceedings of the Symposium on Cloud Computing (SoCC '17)},
  pages 214--228, 2017.

\bibitem{ligra}
Julian Shun and Guy~E. Blelloch.
\newblock Ligra: A Lightweight Graph Processing Framework for Shared Memory.
\newblock In {\em Proceedings of the ACM SIGPLAN Symposium on Principles and
  Practice of Parallel Programming (PPoPP '13)}, pages 135--146, 2013.

\bibitem{tcshun}
Julian Shun and Kanat Tangwongsan.
\newblock Multicore triangle computations without tuning.
\newblock In {\em IEEE International Conference on Data Engineering (ICDE
  '15)}, pages 149--160, 2015.

\bibitem{exemplars}
Qi~Song, Mohammad~Hossein Namaki, and Yinghui Wu.
\newblock Answering Why-Questions for Subgraph Queries in Multi-attributed
  Graphs.
\newblock In {\em IEEE International Conference on Data Engineering (ICDE
  '19)}, pages 40--51, 2019.

\bibitem{distgraph}
Nilothpal Talukder and Mohammed~J. Zaki.
\newblock A Distributed Approach for Graph Mining in Massive Networks.
\newblock {\em Data Mining and Knowledge Discovery}, 30(5):1024--1052, 2016.

\bibitem{arabesque}
Carlos H.~C. Teixeira, Alexandre~J. Fonseca, Marco Serafini, Georgos Siganos,
  Mohammed~J. Zaki, and Ashraf Aboulnaga.
\newblock Arabesque: A System for Distributed Graph Mining.
\newblock In {\em Proceedings of the ACM Symposium on Operating Systems
  Principles (SOSP '15)}, pages 425--440, 2015.

\bibitem{vf2}
Julian Ullmann.
\newblock Bit-vector Algorithms for Binary Constraint Satisfaction and Subgraph
  Isomorphism.
\newblock {\em Journal of Experimental Algorithmics}, 15:1--1, 2011.

\bibitem{kickstarter}
Keval Vora, Rajiv Gupta, and Guoqing Xu.
\newblock {KickStarter}: Fast and Accurate Computations on Streaming Graphs via
  Trimmed Approximations.
\newblock In {\em Proceedings of the International Conference on Architectural
  Support for Programming Languages and Operating Systems (ASPLOS '17)}, pages
  237--251, 2017.

\bibitem{aspire}
Keval Vora, Sai~Charan Koduru, and Rajiv Gupta.
\newblock ASPIRE: Exploiting Asynchronous Parallelism in Iterative Algorithms
  Using a Relaxed Consistency Based DSM.
\newblock In {\em Proceedings of SIGPLAN International Conference on Object
  Oriented Programming Systems Languages and Applications (OOPSLA '14)}, page
  861–878, 2014.

\bibitem{coral}
Keval Vora, Chen Tian, Rajiv Gupta, and Ziang Hu.
\newblock CoRAL: Confined Recovery in Distributed Asynchronous Graph
  Processing.
\newblock In {\em Proceedings of the International Conference on Architectural
  Support for Programming Languages and Operating Systems (ASPLOS '17)}, page
  223–236, 2017.

\bibitem{rstream}
Kai Wang, Zhiqiang Zuo, John Thorpe, Tien~Quang Nguyen, and Guoqing~Harry Xu.
\newblock RStream: Marrying Relational Algebra with Streaming for Efficient
  Graph Mining on a Single Machine.
\newblock In {\em Proceedings of the USENIX Conference on Operating Systems
  Design and Implementation (OSDI '18)}, pages 763--782, 2018.

\bibitem{independence-support}
Yuyi Wang and Jan Ramon.
\newblock An Efficiently Computable Support Measure for Frequent Subgraph
  Pattern Mining.
\newblock In {\em Machine Learning and Knowledge Discovery in Databases (ECML
  PKDD '12)}, pages 362--377, 2012.

\bibitem{fakenews}
Liang Wu and Huan Liu.
\newblock Tracing Fake-News Footprints: Characterizing Social Media Messages by
  How They Propagate.
\newblock In {\em Proceedings of the ACM International Conference on Web Search
  and Data Mining (WSDM '18)}, pages 637--645, 2018.

\bibitem{snap}
Jaewon Yang and Jure Leskovec.
\newblock Defining and Evaluating Network Communities based on Ground-Truth.
\newblock {\em Knowledge and Information Systems}, 42(1):181--213, 2015.

\bibitem{maverick}
Gensheng Zhang, Damian Jimenez, and Chengkai Li.
\newblock Maverick: Discovering Exceptional Facts from Knowledge Graphs.
\newblock In {\em Proceedings of the ACM International Conference on Management
  of Data (SIGMOD '18)}, pages 1317--1332, 2018.

\bibitem{gemini}
Xiaowei Zhu, Wenguang Chen, Weimin Zheng, and Xiaosong Ma.
\newblock Gemini: A Computation-Centric Distributed Graph Processing System.
\newblock In {\em Proceedings of the USENIX Symposium on Operating Systems
  Design and Implementation (OSDI '16)}, pages 301--316, 2016.

\end{thebibliography}

\end{document}